\documentclass[a4paper,onecolumn,10pt]{quantumarticle}
\pdfoutput=1
\usepackage[utf8]{inputenc}
\usepackage[english]{babel}
\usepackage[T1]{fontenc}
\usepackage{amsmath}
\usepackage{amssymb}
\usepackage{amscd}
\usepackage{hyperref}
\usepackage{enumitem}
\usepackage{mathtools}
\DeclarePairedDelimiter\ket{\lvert}{\rangle}
\usepackage[numbers,sort&compress]{natbib}
\usepackage[capitalize]{cleveref}
\usepackage{lineno}
\usepackage{tabularx}  
\usepackage[dvipsnames]{xcolor}
\usepackage{tikz-cd}
\usepackage{lipsum}
\usepackage{subcaption}

\usepackage{framed}

\usetikzlibrary{positioning}
\usetikzlibrary{shapes.geometric, arrows}
\usetikzlibrary{matrix}

\usepackage{amsthm}
\usepackage{arydshln}

\let\mod\relax
\DeclareMathOperator{\mod}{mod}

\newtheoremstyle{break}%
    {}{}%
    {}{}%
    {\bfseries}{}
    {\newline}{}

\theoremstyle{break}
\newtheorem{proposition}{Proposition}[section]

\newtheorem{example}{Example}[section]

\theoremstyle{remark}

\hypersetup{
    colorlinks,
    linkcolor={red!50!black},
    citecolor={green!50!black},
    urlcolor={blue!50!black},
}

\begin{document}
\author{Mark A. Webster}
\email{mark.webster@sydney.edu.au}
\affiliation{Centre for Engineered Quantum Systems, School of Physics,
University of Sydney, Sydney, NSW 2006, Australia}
\affiliation{Sydney Quantum Academy, Sydney, NSW, Australia}

\author{Armanda O. Quintavalle}
\affiliation{Department of Physics and Astronomy, University of Sheffield, Sheffield S3 7RH, United Kingdom}

\author{Stephen D. Bartlett}
\affiliation{Centre for Engineered Quantum Systems, School of Physics,
University of Sydney, Sydney, NSW 2006, Australia}

\title{Transversal Diagonal Logical Operators for Stabiliser Codes}
\begin{abstract}
    Storing quantum information in a quantum error correction code can protect it from errors, but the ability to transform the stored quantum information in a fault tolerant way is equally important. Logical Pauli group operators can be implemented on Calderbank-Shor-Steane (CSS) codes, a commonly-studied category of codes, by applying a series of physical Pauli X and Z gates. Logical operators of this form are fault-tolerant because each qubit is acted upon by at most one gate, limiting the spread of errors, and are referred to as transversal logical operators. Identifying transversal logical operators outside the Pauli group is less well understood. Pauli operators are the first level of the Clifford hierarchy which is deeply connected to fault-tolerance and universality. In this work, we study transversal logical operators composed of single- and multi-qubit diagonal Clifford hierarchy gates. We demonstrate algorithms for identifying all transversal diagonal logical operators on a CSS code that are more general or have lower computational complexity than previous methods. We also show a method for constructing CSS codes that have a desired diagonal logical Clifford hierarchy operator implemented using single qubit phase gates. Our methods rely on representing operators composed of diagonal Clifford hierarchy gates as diagonal XP operators and this technique may have broader applications. 
\end{abstract}

\maketitle
\thispagestyle{plain}

\section{Overview}
Quantum error correction has become a very active area of research because of its potential to mitigate noise in complex quantum devices. Recent experimental results have validated  the storage of quantum information in the codespace of a quantum error correction code as a practical way of protecting it from noise (see \cite{google}, \cite{quantinuum} and \cite{eth}). 
Many of these initial demonstrations have made use of CSS codes~\cite{css}, a well-studied class of quantum error correction codes that are relatively simple to analyse and implement.

To implement algorithms on quantum computers, we also need to transform the stored quantum information in a fault-tolerant way. One method of implementing fault-tolerant logical operations on CSS codes is to use transversal logical operators. Transversal logical operators have depth-one circuit implementations involving single or multi-qubit gates. Such implementations are considered fault-tolerant because an error on one physical qubit can only spread to a limited number of other qubits when applying the logical operator. Whilst the Eastin-Knill theorem rules out the existence of a quantum error correcting code with a set of transversal operators that is universal \cite{eastin_knill}, determining the transversal gates of a quantum error correction code is key to designing a fault-tolerant architecture.

Deeply connected to fault tolerance and universality is the Clifford hierarchy \cite{clifford_hierarchy} of unitary operators. The first level of the Clifford hierarchy is the Pauli group $\mathcal{CH}_1 := \langle iI, X, Z\rangle$. Conjugation of Paulis by operators at level $t+1$ results in an operator at level $t$. The level $t+1$ operators $A \in \mathcal{CH}_{t+1}$ are then defined recursively as those for which $ABA^{-1} \in \mathcal{CH}_t$ for all $B \in \mathcal{CH}_1$. Level $2$ Clifford hierarchy gates include the single-qubit Hadamard and $S:=\sqrt{Z}$ gates, as well as the 2-qubit controlled-Z ($CZ$) gates. Level $3$ gates include the single-qubit $T:=\sqrt{S}$ gate as well as the multi-qubit controlled-S ($CS$) and controlled-controlled-Z ($CCZ$) gates.  A set of gates that includes all level-2 gates and at least one level-3 gate is universal \cite{clifford_universality}.

Logical Pauli group operators can be implemented transversally on CSS codes and identifying these is relatively straightforward. Identifying transversal logical operators at higher levels of the Clifford hierarchy is more challenging and existing methods are of exponential complexity in either the number of physical or logical qubits in the code. Some classes of CSS codes with high degrees of symmetry are known to have non-Pauli transversal logical operators. Examples using single-qubit diagonal gates include the 7-qubit Steane code \cite{steane_code},  two-dimensional color codes \cite{2d_colour} and triorthogonal codes \cite{triorthogonal}. Examples of CSS codes which have logical operators made from single and multi-qubit gates include the  two-dimensional toric code  \cite{moussa}, codes with ZX-symmetries \cite{fold_transversal} and symmetric hypergraph product codes \cite{partitioning_qubits}. 

In this paper, we present a suite of method and algorithms for identifying diagonal transversal logical operators on any CSS code, without any knowledge of any symmetries of the code. The building blocks of our logical operators are physical single- or multi-qubit diagonal gates, at a given level $t$ of the Clifford hierarchy. Our methods scale as a polynomial in the number of physical and/or logical qubits in the code, with one exception. We also give a method for constructing a CSS code that has a transversal implementation of a desired diagonal logical Clifford hierarchy operator using single-qubit gates. Our new algorithms use the XP formalism, introduced in Ref.~\cite{XP}, which is a powerful tool for representing the logical operator structure of a stabiliser code.   

\subsection{Existing Work on Transversal Logical Operators}
We briefly review previous methods for identifying diagonal logical operators of arbitrary CSS codes, and methods for constructing CSS codes with a desired transversal logical operator.
In Ref.~\cite{GF4}, a method is given to find all logical operators at level 2 of the Clifford hierarchy for a CSS code by mapping it to a classical code over $GF(4)$. This method involves calculating the automorphism group of the classical code, which has exponential complexity in the number of qubits in the stabiliser code \cite{automorphism_complexity}. 

There has also been a significant amount of work on logical operators constructed from single- and multi-qubit diagonal Clifford hierarchy gates. In Ref.~\cite{unifying_clifford}, operators composed of diagonal Clifford hierarchy gates on one or two qubits are shown to be representable as symmetric matrices over $\mathbb{Z}_N$, referred to as Quadratic Form Diagonal (QFD) gates. Necessary and sufficient conditions for a QFD gate to act as a logical operator on a CSS code are then presented. 
In Ref.~\cite{synthesis}, a method of generating circuits using multi-qubit gates which implement arbitrary logical operators at level 2 of the Clifford hierarchy is presented. A method for generating CSS codes with transversal diagonal logical operators at increasing levels of the Clifford hierarchy is presented in Ref.~\cite{climbing}, along with a method to increase the $Z$-distance of such codes. In Ref.~\cite{XP}, we demonstrated an algorithm for finding all diagonal logical operators composed of single-qubit phase gates which, for CSS codes, involves taking the kernel modulo $N$ of a matrix with $n+ 2^k$ columns where $n$ and $k$ are the number of physical and logical qubits respectively.

\subsection{Contribution of this Work}
In this work, we present efficient methods to identify and test diagonal logical operators on CSS codes using both single and multi-qubit diagonal Clifford hierarchy gates as building blocks. 
These methods generalise to non-CSS stabiliser codes.  We also present a technique for generating CSS codes with implementations of any desired diagonal Clifford hierarchy logical operator using single-qubit phase gates. 


We first consider operators composed of single-qubit phase gates at level $t$ of the Clifford hierarchy. We show that these can be represented as diagonal XP operators of precision $N=2^t$. For logical operators of this form, we demonstrate the following algorithms that apply to any CSS code and at any desired level of the Clifford hierarchy:
\begin{enumerate}
    \item \textbf{Finding a generating set of diagonal logical identity operators for the code}: An XP operator may act as a logical identity, but may not be an element of the stabiliser group of a CSS code. The logical identities are used as inputs to several other algorithms (\Cref{sec_LI})
    \item \textbf{Search for an implementation of a desired logical controlled-phase operator on the code}: useful for checking if a given CSS code has a transversal implementation of a particular logical operator and for checking the results of other algorithms (\Cref{sec_ker_search});
    \item \textbf{Determining if a given diagonal operator acts as a logical operator on the code}: This method is of linear complexity in the number of independent $X$-checks whereas existing methods are of exponential complexity (\Cref{sec_LO_test});
    \item \textbf{Finding a generating set of diagonal logical operators on the code}: The generating set gives us a complete understanding of the diagonal logical operator structure of a CSS code, and can be used on CSS codes with a large number of physical and logical qubits at any desired level of the Clifford hierarchy (\Cref{sec_comm_LO});
    \item \textbf{Expressing the action of a diagonal logical operator as a product of logical controlled-phase gates}: The action of a logical operator can be difficult to interpret, particularly for codes with a large number of logical qubits. This method greatly simplifies the interpretation of logical actions (\Cref{sec_logical_action}).
\end{enumerate}
We then show that multi-qubit diagonal Clifford hierarchy gates acting on a codesepace can be represented as diagonal XP operators acting on a larger Hilbert space via an embedding operator (\Cref{sec_embedded_codes}). We demonstrate algorithms for:
\begin{enumerate}
\setcounter{enumi}{5}
    \item \textbf{Finding depth-one implementations of logical operators composed of diagonal Clifford hierarchy gates}: on small CSS codes, this allows us to identify and verify the depth-one logical operators of \cite{moussa,fold_transversal,partitioning_qubits} with no knowledge of the symmetry of the code (\Cref{sec_depth_one_algorithm});
    \item \textbf{Canonical implementations of a desired logical controlled-phase operator composed of multi-qubit controlled-phase gates}: this allows us to write closed-form expressions for arbitrary diagonal Clifford hierarchy logical operators \Cref{sec_canonical_logical_operators};
    \item \textbf{Construction of CSS codes which have an implementation of a desired logical controlled-phase operator composed of single qubit phase gates}: the canonical logical operator implementation allows us to construct families of CSS codes which have transversal implementations of a desired diagonal Clifford hierarchy logical operator \Cref{sec_constructing_css_codes}.
\end{enumerate}
Apart from the depth-one search algorithm, the eight algorithms have complexity that is polynomial in the parameters $n,k,r$ of the CSS code (see below). As a result, they can be applied to `large' codes that have so far been out of reach of existing methods. There are no restrictions on the level of the Clifford hierarchy or maximum support size of the physical gates used in the methods.  

A summary of the characteristics and computational complexity of search and test algorithms is presented in \Cref{tab_alg_summary}. Complexity is expressed in terms of the following variables:
\begin{itemize}
    \item Required level of the Clifford hierarchy $t$;
    \item Number of physical qubits $n$ in the CSS code;
    \item Number of logical qubits $k$ in the CSS code;
    \item Number of independent $X$-checks $r$ in the CSS code;
\end{itemize}
The space complexity of the algorithm is expressed in terms of the size of the key matrices used. The time complexity is expressed in terms of the number of kernel operations performed on the key matrices - these operations dominate the complexity of the algorithms.

The algorithms have been implemented in a Python \href{https://github.com/m-webster/CSSLO}{GitHub repository} accessible under the GNU General Public License. A range of sample codes are also available for testing in this repository, including Reed-Muller codes, hyperbolic surface codes, triorthogonal codes and symmetric hypergraph product codes.

\begin{table}[t]
	\centering
\begin{tabularx}{0.97\textwidth} { 
		| >{\raggedright\arraybackslash}p{5cm} 
		|| >{\raggedright\arraybackslash}X 
		| >{\raggedright\arraybackslash}X |
		| >{\raggedright\arraybackslash}X | }
	\hline
	 \textbf{Algorithm} & \textbf{Gate Type} & \textbf{Space Complexity}&\textbf{Time Complexity}\\
	 \hline
  	\hline
   	\textbf{1. Diagonal Logical Identity Group Generators} & Single-qubit  & $\mathcal{O}((k+r)^t\times n)$ & $\mathcal{O}(1)$\\
      
	  \hline
        \textbf{2. Search by Logical Action} & Single-qubit  & $\mathcal{O}((k+r)^t\times n)$ & $\mathcal{O}(1)$\\
  \hline
	 \textbf{3. Logical Operator Test\textsuperscript{*}} & Single-qubit  & $\mathcal{O}(1\times n)$ & $\mathcal{O}(r)$\\
	  \hline
   	\textbf{4. Diagonal Logical Operator Group Generators\textsuperscript{*}} & Single-qubit  & $\mathcal{O}(n \times n)$ & $\mathcal{O}(r)$\\
    	\hline
   	\textbf{5. Determine Action of Diagonal Logical Operator} & Single-qubit  & $\mathcal{O}((k+r)^t \times n)$ & $\mathcal{O}(1)$\\
        \hline
   	\textbf{6. Depth-One Logical Operators\textsuperscript{**}} & Multi-qubit  & $\mathcal{O}(n^t \times n^t)$ & $\mathcal{O}(2^n)$\\
    \hline
\end{tabularx}
\caption{Comparison of Search and Test Algorithms for Diagonal Logical Operators. Note that entries annotated with {*} require the diagonal logical identities of algorithm 1 as input. Entries annotated with {**} require the diagonal logical operators of algorithm 4 as input.}
\label{tab_alg_summary}
\end{table}

\section{Background}

This Section reviews the necessary background material for this work. We first introduce the Clifford hierarchy of diagonal operators and introduce a vector representation of these. We then outline notation and fundamental properties of CSS codes. Next, we define what we mean by a diagonal logical operator on a CSS code. We then present an example illustrating the types of diagonal logical operators we consider in this work for the well-known $[[4,2,2]]$ code. We then review the XP stabiliser formalism and some fundamental properties of the XP operators, which we will use to represent logical operators composed of diagonal Clifford hierarchy gates. We explain the logical operator group structure in the XP formalism, which is somewhat different than in the Pauli stabiliser formalism. 

\subsection{Diagonal Clifford Hierarchy Operators}\label{sec_controlled_phase_operators}
Here we review the properties of operators in the diagonal Clifford hierarchy. We will use diagonal gates at level $t$ of the  Clifford hierarchy on $n$ qubits as the building blocks for logical operators.  
The diagonal Clifford hierarchy operators at each level form a group generated by the following operators~\cite{cui}:
\begin{itemize}
    \item Level 1: Pauli $Z$ gate on qubit $i: 0 \le i < n$ denoted $Z_i$;
    \item Level 2: Controlled-$Z$ ($CZ_{ij}$) and $S_i:= \sqrt{Z_i}$;
    \item Level 3: $CCZ_{ijk}, CS_{ij}$ and $T_i:= \sqrt{S_i}$;
    \item Level $t+1$: Square roots and controlled versions of operators from level $t$.
\end{itemize}

At each level, we refer to the generators as \textbf{level-$t$ controlled-phase gates}. Where an operator is an element of the diagonal Clifford hierarchy group at level $t$, we say that it is \textbf{composed of level-$t$ controlled-phase gates}. 

The single-qubit \textbf{phase gate} at  level $t$ is of form $\text{diag}(1, \exp(2\pi i /N))$ where $N:=2^t$. If an operator is an element of the group generated by single-qubit phase gates at level $t$, we say it is \textbf{composed of level-$t$ phase gates}.

The matrix form of any diagonal transversal logical operator of a CSS code must have entries of form $\exp(q\pi i /2^t)$ for integers $q, t$, as shown in Ref.~\cite{anderson2014classification}. Such matrices are elements of the diagonal Clifford hierarchy group at some level, and so considering logical operators composed of controlled-phase gates yields all possible diagonal transversal logical operators on a CSS code.

\subsection{Vector Representation of Controlled-Phase Operators}\label{sec_CP_vector}
We now introduce a vector representation of controlled-phase operators that underpins our analytical methods. Fix a level $t$ of the Clifford hierarchy (\Cref{sec_controlled_phase_operators}) and let $N:= 2^t$. Let $\omega:= e^{\pi i/N}$ be a $(2N)$-th root of unity. The operator $\text{CP}_N(q,\mathbf{v})$, where $q \in \mathbb{Z}_{2N}$ and $\mathbf{v}$ is a binary vector of length $n$, is defined as follows by its action on a computational basis vectors $|\mathbf{e}\rangle$ for $\mathbf{e} \in \mathbb{Z}_2^n$:
\begin{align}
    \text{CP}_N(q,\mathbf{v})|\mathbf{e}\rangle := \begin{cases}
        \omega^{q}|\mathbf{e}\rangle &\text{ if }\mathbf{v} \preccurlyeq \mathbf{e};\\
        |\mathbf{e}\rangle & \text{ otherwise}.\label{eq_CP_def}
    \end{cases}
\end{align}
The relation $\preccurlyeq$ is a partial order for binary vectors based on their support (the set of indices where the vector is non-zero). The expression $\mathbf{v} \preccurlyeq \mathbf{e}$ indicates $\text{supp}(\mathbf{v}) \subseteq \text{supp}(\mathbf{e})\iff \mathbf{e}\mathbf{v} = \mathbf{v}$ where vector multiplication is componentwise. For an integer $0 \le i < n$, we will also write $i \preccurlyeq \mathbf{v}$ if $\mathbf{v}[i] = 1$. The phase applied can be expressed more concisely as follows:
\begin{align}
    \text{CP}_N(q,\mathbf{v})|\mathbf{e}\rangle = \omega^{q\cdot p_\mathbf{v}(\mathbf{e})}|\mathbf{e}\rangle\text{ where }p_\mathbf{v}(\mathbf{e}) := \prod_{i \preccurlyeq \mathbf{v}}\mathbf{e}[i].\label{eq_CP_phase}
\end{align}
Each generator of the diagonal Clifford hierarchy can be written in vector form. To see this, we note that the phase gate at level $t$ can be written as $P:= \text{diag}(1,\omega^2)$. The phase operator acting on qubit $i$ can be written in vector form as $P_i = \text{CP}_N(2,\mathbf{b}^n_i)$ where $\mathbf{b}^n_i$ is the length $n$ binary vector, which is all zero apart from component $i$ which is one. Similarly, the operator $CP_{ij} = \text{CP}_N(2,\mathbf{b}^n_{ij})$ where $\mathbf{b}^n_{ij}$ is zero apart from components $i$ and $j$. 
 The operators of form $\text{CP}_N(2^{\text{wt}(\mathbf{v})},\mathbf{v})$ with $1 \le \text{wt}(\mathbf{v}) \le t$ are the generators of the level-$t$ controlled-phase operators presented in \Cref{sec_controlled_phase_operators}. 
\begin{example}[Vector Representation of Level $3$ Controlled-Phase Operators]
This example illustrates the vector representation of level $3$ diagonal Clifford hierarchy operators. At level $t=3$ the generators have vector representations as follows:
\begin{align}
T_{i} &= CP_8(2,\mathbf{b}^n_{i})\\
CS_{ij} &= CP_8(4,\mathbf{b}^n_{ij})\\
CCZ_{ijk} &= CP_8(8,\mathbf{b}^n_{ijk}).
\end{align}
We also include $\omega I = CP_8(1,\mathbf{0})$ as a generator at the third level of the hierarchy as phases of this form occur in the commutation relation for controlled-phase operators - see \Cref{eq_cpXi,eq_CP_relation}.
\end{example}

\subsection{CSS Codes}\label{sec_css_codes}
Here we introduce some key notation and results for CSS codes. Our notation for CSS codes is somewhat different to that in the literature and is used because it simplifies the statement of our results. Although we focus on CSS codes in this work, the methods are applicable to any stabiliser code as set out in \Cref{app_non_css}. For our purposes, a CSS code on $n$ qubits is specified by an $r \times n$ binary matrix $S_X$ the rows of which we refer to as the \textbf{$X$-checks} and a $k \times n$ binary matrix $L_X$ whose rows are referred to as the \textbf{$X$-logicals}. We assume that the rows of $S_X$ and $L_X$ are independent binary vectors - otherwise we can use linear algebra modulo $2$ to ensure this. The \textbf{$Z$-checks} can be calculated by taking the kernel modulo $2$ of the $X$-checks and $X$-logicals, i.e., 
\begin{align}
S_Z := \ker_{\mathbb{Z}_2}\begin{pmatrix}S_X\\L_X\end{pmatrix}.\label{eq_Z-checks}
\end{align}
In \Cref{eq_Z-checks}, the notation $\ker_{\mathbb{Z}_2}$ refers to the basis in reduced row echelon form of the kernel modulo $2$ of a binary matrix. We form \textbf{stabiliser generators} $\mathbf{S}_X, \mathbf{S}_Z$ from the rows of $S_X$ and $S_Z$ in the obvious way - if $\mathbf{x}$ is a row of $S_X$ then the corresponding stabiliser generator is $\prod_{0 \le i < n}X_i^{\mathbf{x}[i]}$. The \textbf{codespace} is the simultaneous $+1$ eigenspace of the \textbf{stabiliser group} $\langle \mathbf{S}_X, \mathbf{S}_Z \rangle$ and is a subspace of $\mathcal{H}_2^n$. The codespace is spanned by $2^k$ \textbf{canonical codewords} which are indexed by binary vectors $\mathbf{v}$ of length $k$ and are defined as follows:
\begin{align}
|\mathbf{v}\rangle_L := \sum_{\mathbf{u} \in \mathbb{Z}_2^r}|\mathbf{e_{uv}}\rangle := \sum_{\mathbf{u} \in \mathbb{Z}_2^r}|\mathbf{u}S_X + \mathbf{v}L_X\rangle.\label{eq_codewords}
\end{align}
In the above expression, matrix operations are modulo $2$. For simplicity, we are not concerned with normalising codeword states. It may be possible to make a different \textbf{choice of basis} for the span $\langle L_X \rangle$ over $\mathbb{Z}_2$. The choice of basis affects the labelling of the canonical codewords by binary vectors $\mathbf{v}$ of length $k$, but does not otherwise change the set of canonical codewords. 

\subsection{Logical Operators of CSS Codes}\label{sec_LO_of_CSS_codes}
We now describe what we mean by a logical operator on a CSS code. Let $\mathcal{C}:\mathcal{H}_2^k \rightarrow \mathcal{H}_2^n$ be the \textbf{encoding operator} which takes computational basis vectors to canonical codewords of \Cref{eq_codewords} i.e. $\mathcal{C}\ket{\mathbf{v}} =\ket{\mathbf{v}}_L$ for $\mathbf{v}\in \mathbb{Z}_2^k$. Now let $B$ be a unitary operator acting on $k$ qubits. We say that an operator $\overline{B}$ acting on $n$ qubits is a \textbf{logical $B$ operator} if 
\begin{align} \overline{B}\mathcal{C} = \mathcal{C}B.\label{eq_LO_condition}\end{align} 
A unitary operator $B$ is \textbf{diagonal} if we can write $B:= \text{diag}(\mathbf{c})$ for some complex-valued vector $\mathbf{c}$ of length $2^k$ representing the phase applied to each computational basis vector, i.e. $B\ket{\mathbf{v}} = \mathbf{c_v}\ket{\mathbf{v}}$ for $\mathbf{v} \in \mathbb{Z}_2^k$ and $\mathbf{c_v}\in \mathbb{C}$.
If $\overline{B}$ is a diagonal logical operator, then $B$ is diagonal as well, though the converse is not necessarily true. 
From \Cref{eq_LO_condition} and \Cref{eq_codewords}, we have:
\begin{align}
    \overline{B}\mathcal{C}\ket{\mathbf{v}} &= \overline{B}\ket{\mathbf{v}}_L= \overline{B}\sum_{\mathbf{u} \in \mathbb{Z}_2^r}|\mathbf{e_{uv}}\rangle= \sum_{\mathbf{u} \in \mathbb{Z}_2^r}\overline{B}|\mathbf{e_{uv}}\rangle\\
    &= \mathcal{C}B\ket{\mathbf{v}} = \mathbf{c_v}\ket{\mathbf{v}}_L = \sum_{\mathbf{u} \in \mathbb{Z}_2^r}\mathbf{c_v} |\mathbf{e_{uv}}\rangle.\label{eq_LO_CSS}
\end{align}
As a result, we can check if $\overline{B}$ is a logical $B$ operator by doing the following:
\begin{enumerate}
\item For each $\mathbf{v}\in \mathbb{Z}_2^k$, calculate $\mathbf{c_v} \in \mathbb{C}$ such that $B \ket{\mathbf{v}} = \mathbf{c_v}\ket{\mathbf{v}}$;
\item For each $\mathbf{u} \in \mathbb{Z}_2^r$, check that $\overline{B} \ket{\mathbf{e_{uv}}} = \mathbf{c_v}\ket{\mathbf{e_{uv}}}$.
\end{enumerate}
This method of checking whether a diagonal unitary is a logical operator involves $\mathcal{O}(2^{r+k})$ steps; we present a method in \Cref{sec_LO_test} with linear complexity in $r$. 

We say that an operator $\overline{B}$ is a \textbf{logical identity} if $\overline{B}\ket{\mathbf{v}}_L = \ket{\mathbf{v}}_L$ for all $\mathbf{v}\in \mathbb{Z}_2^k$ - that is, it fixes each canonical codeword and hence each element of the codespace. If $\overline{B}$ is diagonal, as a consequence of \Cref{eq_LO_CSS}, it is a logical identity if and only if $\overline{B}\ket{\mathbf{e_{uv}}} =\ket{\mathbf{e_{uv}}}$ for all $\mathbf{u}\in \mathbb{Z}_2^r, \mathbf{v}\in \mathbb{Z}_2^k$.

Whether a diagonal operator is a logical identity or a logical operator is independent of the choice of basis for the span $\langle L_X\rangle$ (see \Cref{sec_css_codes}). However, the logical action of the operator depends on the labelling the canonical codewords and so is dependent on the choice of basis for $\langle L_X\rangle$.

\begin{example}[Transversal Logical Operators of {[[4,2,2]]} Code]\label{eg_CSS_LO}
We use the [[4,2,2]] code to illustrate  the  types of transversal logical operators we consider in this work. Using the notation introduced in \Cref{sec_css_codes}, the $X$-checks and $X$-logicals of the code are:
\begin{align}
    S_X &:= \begin{pmatrix}1111\end{pmatrix}\\
    L_X &:= \begin{pmatrix}0101\\0011\end{pmatrix}
\end{align}
In this case, there are $r=1$ $X$-checks and $k=2$ $X$-logicals. 
There are $2^k = 4$ canonical codewords which we calculate using \Cref{eq_codewords}:
\begin{align}
\begin{matrix}
|00\rangle_L &:= |0000\rangle + |1111\rangle \\
|01\rangle_L &:= |0011\rangle + |1100\rangle \\
|10\rangle_L &:= |0101\rangle + |1010\rangle \\
|11\rangle_L &:= |0110\rangle + |1001\rangle 
\end{matrix}
\end{align}
We can calculate the single $Z$-check as follows:
\begin{align}
    S_Z &:= \ker_{\mathbb{Z}_2}\begin{pmatrix}
        S_X\\L_X
    \end{pmatrix} = \begin{pmatrix}
        1111
    \end{pmatrix}.
\end{align}
Readers can verify that $Z^{\otimes 4}$ acts as a logical identity by checking that $Z^{\otimes 4}\ket{\mathbf{v}}_L = \ket{\mathbf{v}}_L$ for each of the canonical codewords. 

The following are examples of transversal diagonal logical operators composed of controlled-phase gates at level $2$ whose actions can be verified by applying the method of \Cref{sec_LO_of_CSS_codes}:
\begin{enumerate}
	\item \textbf{Single-qubit phase gates} Controlled-Z:  $\overline{CZ_{01}} = S_0^3S_2S_2S_3^3$
	\item \textbf{Multi-qubit controlled-phase gates} S operator on both logical qubits:  $\overline{S_0S_1}= S_1S_2CZ_{03}$
\end{enumerate}
\end{example}

\subsection{The XP Formalism}\label{sec_XP}
The XP formalism is a generalisation of the Pauli stabiliser formalism, and we will show that diagonal Clifford hierarchy operators can be represented as diagonal XP operators. In the XP formalism, we fix an integer \textbf{precision} $N \ge 2$ and let $\omega = \exp(\pi i/N)$ be a $(2N)$-th root of unity. We define a diagonal phase operator $P = \text{diag}(1, \omega^2)$ which is a $1/N$ rotation around the Z axis and consider the group of XP operators $\mathcal{XP}^n_N$ that is generated by $\omega I, X_i, P_i$ where $P_i$ is a $P$ operator applied to qubit $i$. By setting $N:=2^t$, it is easy to see that the $P_i$ correspond to the level $t$ phase gates of \Cref{sec_controlled_phase_operators}, and so any operator composed of single-qubit phase gates can be represented as a diagonal XP operator. For example, setting $t=1$ results in $N=2, \omega = i$ and $P = Z$ so $\mathcal{XP}_2^n$ is the Pauli group on $n$ qubits.

The XP formalism has a fundamental commutation relation that allows us to move $P$ operators to the right of $X$ operators:
\begin{align}
    PX = \omega^2XP^{-1}. \label{eq_XP_relation}
\end{align}
All XP operators have a unique \textbf{vector representation} with a phase component $p\in \mathbb{Z}_{2N}$, an $X$-component $\mathbf{x} \in \mathbb{Z}_2^n$ and a $Z$-component $\mathbf{z} \in \mathbb{Z}_N^n$. 
The $Z$-component is modulo $N$, for instance, because $P^N = I$.
The XP operator formed from these components is:
\begin{align}
XP_N(p|\mathbf{x}|\mathbf{z}) := \omega^p \prod_{0 \le i < n}X_i^{\mathbf{x}[i]}P_i^{\mathbf{z}[i]}.\label{eq_XP_def}
\end{align}
\textbf{Diagonal XP operators} are those with a zero $X$-component. 
The vector form of XP operators allows us to perform algebraic operations efficiently via componentwise addition and multiplication of vectors - examples are given in Table 4 of \cite{XP}. In particular, the \textbf{action of an XP operator} on a computational basis element $|\mathbf{e}\rangle$ where $\mathbf{e} \in \mathbb{Z}_2^n$ is determined as follows:
\begin{align}
XP_N(p|\mathbf{x}|\mathbf{z})|\mathbf{e}\rangle = \omega^{p 
 + 2 \mathbf{e} \cdot \mathbf{z}}|\mathbf{e} \oplus \mathbf{x}\rangle.\label{eq_XP_action}
\end{align}
Where $N=2^t$, we can determine the lowest level of the Clifford hierarchy at which a diagonal operator $B := XP_N(0|\mathbf{0|z})$ occurs. Let $g := GCD(N, \mathbf{z})$ be the GCD of $N$ and each component of $\mathbf{z}$. As $N=2^t, g$ is a power of $2$ and $B = XP_{N/g}(p/g|\mathbf{0|z}/g)$. Accordingly, $B$ occurs at level $t - \log_2(g)$ of the diagonal Clifford hierarchy.

\begin{example}[Determining Clifford Hierarchy Level of XP operators]
Let $t=3$ and $B = XP_8(0|\mathbf{0}|4444) $, so that $g = GCD(8,4) = 4$. Hence $B = XP_2(0|\mathbf{0}|1111)= Z^{\otimes 4}$ and occurs at level $t - \log_2(4) = 3-2= 1$ of the Clifford hierarchy.
\end{example}

\subsection{Logical Identity and Logical Operator Groups in the XP Formalism}\label{sec_XP_groups}
We now look at the logical group structure of a CSS code in the XP formalism with reference to the definitions of logical operators in  \Cref{sec_LO_of_CSS_codes}. 
In the stabiliser formalism, a Pauli operator acts as a logical identity if and only if it is in the stabiliser group $\langle \mathbf{S}_X,\mathbf{S}_Z\rangle$. 
In the XP stabiliser formalism, an XP operator may act as a logical identity but not be in the stabiliser group - we will see an instance of this in \Cref{eg_hypercube_LI}. 
The \textbf{logical XP identity group}, $\mathcal{I}_\text{XP}$, are the XP operators of precision $N$ which fix each element of the codespace. The stabiliser group is a subgroup of $\mathcal{I}_\text{XP}$ but may not be equal to it.

The \textbf{logical XP operator group}, $\mathcal{L}_\text{XP}$, are the XP operators of precision $N$ that are logical $B$ operators for some unitary $B$ acting on $k$ qubits. 
Logical XP operators may have actions outside the Pauli group, and the logical $\overline{CZ_{01}}$ operator of \Cref{eg_CSS_LO} is an instance of such an operator. Logical identities are elements of $\mathcal{L}_\text{XP}$ that have a trivial action. 
The logical groups in the XP formalism are summarised in \Cref{fig_groupdiagram}.
\begin{figure}[hbt!]
\centering
\begin{tikzpicture} 
[set/.style = {thick,
    fill=Rhodamine,
    opacity = 0.2,
    text opacity = 1}]
    
\filldraw[set] (0,1) circle (1);
\filldraw[set] (0.5,1) ellipse[x radius=2, y radius=1.5];
\filldraw[set] (1,1) ellipse[x radius=3, y radius=2];
\draw (-2.5,-1.5) rectangle (5,3.5) ;
\draw (4,-1) node {$\mathcal{XP}^n_N$};
\draw (3,1) node  {$\mathcal{L}_\text{XP}$};
\draw (1.5,1) node  {$\mathcal{I}_\text{XP}$};
\draw (0,1) node  {$\langle \mathbf{S}_X,\mathbf{S}_Z  \rangle$};
\end{tikzpicture}
\caption{\textbf{Relationship between XP Operator Groups}:  Here, $\mathcal{XP}^n_N$ is the group of all XP operators of precision $N$ on $n$ qubits.  The stabiliser group $\langle \mathbf{S}_X,\mathbf{S}_Z  \rangle$ of a CSS code is a subgroup of the logical XP identity group $\mathcal{I}_\text{XP}$ which fixes all elements of the codespace which, in turn, is a subgroup of the logical operators of XP form $\mathcal{L}_\text{XP}$.}\label{fig_groupdiagram}
\end{figure}
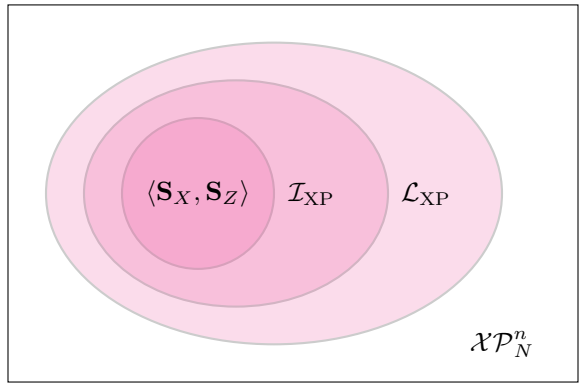

\section{Logical Operators Composed of Single-Qubit Phase Gates}
In this Section, we present methods for identifying and testing logical operators composed of single-qubit phase gates at a given level $t$ of the Clifford hierarchy. 
Operators of this form can be identified with diagonal XP operators of precision $N=2^t$. 
The algorithms in this Section are of polynomial complexity in the code parameters $n,k,r$ (\Cref{sec_css_codes}), so they can be used on CSS codes with a large number of physical or logical qubits.

This Section is structured as follows.  We first show how to calculate generators for the diagonal logical identity XP group. This is an important first step for a number of our algorithms.
We then demonstrate an algorithm that searches for a diagonal XP operator with a desired logical action. Next, we set out an efficient method for testing if a given diagonal XP operator is a logical operator on a CSS code. We then show how to use this test to find all diagonal logical operators of XP form. Finally, we show how to express the action of a diagonal logical XP operator in terms of a product of logical controlled phase operators.
We use the Hypercube code of \cite{hypercube,Kubica1} which has a rich logical operator structure an example throughout this Section. We also demonstrate the use of the algorithms on larger codes such as hyperbolic color codes \cite{LDPC_colour}, poset codes \cite{poset_codes} and triorthogonal codes \cite{triorthogonal}.

\subsection{Diagonal Logical XP Identity Group Generators}\label{sec_LI}
Calculating generators for the logical identity group of a CSS code is an important first step for several of the algorithms discussed in this paper. An algorithm for determining the logical identity group is set out in Section 6.2 of \cite{XP}. Here, we present a simplified version for CSS codes.


Due to the discussion in \Cref{sec_LO_of_CSS_codes}, a diagonal logical identity operator fixes all $\ket{\mathbf{e_{uv}}}$ in the canonical codewords of \Cref{eq_codewords}.  
Now let $N:=2^t$ and let $B:=XP_N(2p|\mathbf{0}|\mathbf{z})$ be a diagonal XP operator. 
Using \Cref{eq_XP_action}, the action of $B$ on the computational basis vector $\ket{\mathbf{e_{uv}}}$ is $B|\mathbf{e_{uv}}\rangle = \omega^{2p + 2\mathbf{e_{uv}}\cdot \mathbf{z}}|\mathbf{e_{uv}}\rangle$. 
Considering the action of $B$ on $\ket{\mathbf{e_{00}}} = \ket{\mathbf{0}}$, we see that $p=0 \mod 2N$.
As $\omega^{2N} =1$, $B$ applies a trivial phase to $\ket{\mathbf{e_{uv}}}$ if and only if $\mathbf{e_{uv}}\cdot \mathbf{z} = 0 \mod N$. 
We can find all such solutions by taking the kernel of a suitably constructed matrix modulo $N$. 
This is done via the \textbf{Howell matrix form} \cite{howell} which is a generalisation of the reduced row echelon form for modules over rings such as $\mathbb{Z}_N$. 
The notation $\ker_{\mathbb{Z}_N}(E_M)$ means the Howell basis of the kernel of the matrix $E_M$ modulo $N$.

\begin{framed}
\vspace{-0.25cm}
\paragraph{Algorithm 1: Logical Identity Group Generators}
\paragraph{Input:}
\begin{enumerate}
    \item The $X$-checks $S_X$ and $X$-logicals $L_X$ of a CSS code (\Cref{sec_css_codes});
    \item The desired level of the Clifford hierarchy $t$ (\Cref{sec_controlled_phase_operators}).
\end{enumerate}
\paragraph{Output:}
A matrix $K_M$ whose rows are the $Z$-components of a set of generators for the diagonal logical identity XP group of precision $N=2^t$ (\Cref{sec_XP_groups}).
\paragraph{Method:}
\begin{enumerate}
    \item Let $E_M$ be the binary matrix whose rows are the $\mathbf{e_{uv}}:= \mathbf{u}S_X + \mathbf{v}L_X$ of \Cref{eq_codewords};
    \item Let $N:= 2^t$ and calculate $K_M:=\ker_{\mathbb{Z}_N}(E_M)$ in Howell matrix form;
    \item Return $K_M$
\end{enumerate}
\end{framed}

Because $E_M$ has $2^{r + k}$ rows, the complexity of the logical identity algorithm is highly sensitive to the number of $X$-checks $r$ and logical qubits $k$. However, due to Proposition E.13 of Ref.~\cite{XP}, we only need to consider $\mathbf{e_{uv}}$ where $\text{wt}(\mathbf{u}) + \text{wt}(\mathbf{v}) \le t$ to determine the logical identity group up to level $t$ of the Clifford hierarchy. Hence, we only require $\genfrac[]{0pt}{2}{r+k}{t} := \sum_{0 \le j \le t}\genfrac(){0pt}{2}{r+k}{j}$ rows from $E_M$ and so the dimensions of the matrix and scale as a polynomial of degree $t$ in $r+k$. 

\begin{example}[Logical Identity Algorithm - Hypercube Code]\label{eg_hypercube_LI}
\begin{figure}[hbt!]
\centering
\includegraphics[width=0.9\textwidth]{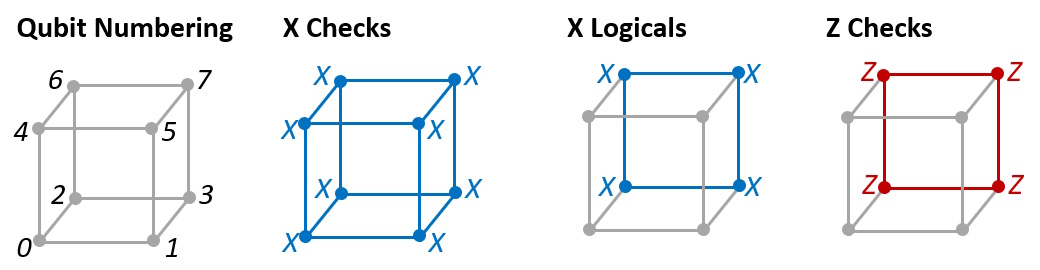}
\caption{\textbf{Hypercube Code of Dimension $3$}: qubits reside on the vertices of a cube. The blue-coloured $X$-logicals are associated with the 2D faces, whilst the $X$-check is associated with the single 3D volume. The red-coloured $Z$-checks are associated with the 2D faces.}\label{fig_hypercube_SXLXSZ}
\end{figure}

In this example, based on Refs \cite{Kubica1} and \cite{hypercube},  qubits reside on the eight vertices of a cube. 
The single $X$-check is the all-ones vector indicating an X operator on all vertices of the cube: 
\begin{align}S_X =\begin{pmatrix}11111111\end{pmatrix}.\end{align}
The three $X$-logicals are weight 4 vectors associated with three faces meeting at a point which we write in the notation of \Cref{sec_css_codes} as follows: 
\begin{align}
L_X = \begin{pmatrix}01010101\\00110011\\00001111\end{pmatrix}.
\end{align}
We calculate the $Z$-checks by applying \Cref{eq_Z-checks} and find that the $Z$-checks also correspond to faces:
\begin{align}
    S_Z := \ker_{\mathbb{Z}_2}\begin{pmatrix}S_X\\L_X\end{pmatrix} = \begin{pmatrix}10010110\\01010101\\00110011\\00001111\end{pmatrix}.
\end{align}
This process is exactly the same as finding the diagonal logical identities at level $t=1$ as outlined in \Cref{sec_LI}. In this case, $E_M$ has $r+k = 1+3 = 4$ rows and the logical identities are the kernel of $E_M$ modulo $2$. Now applying the logical identity algorithm at level $t=3$, $E_M$ has $15$ rows representing the sum modulo 2 of up to $3$ rows from $S_X$ and $L_X$. Taking the kernel of $E_M$ modulo $N=2^3 = 8$, we find:
\begin{align}
    K_M := \ker_{\mathbb{Z}_8}(E_M) = \begin{pmatrix}
         22222222\\
04040404\\
00440044\\
00004444
    \end{pmatrix}
\end{align}
The rows of $K_M$ are the $Z$-components of diagonal XP operators which act as logical identities, and form a generating set of all such operators of precision $N$. For instance, the operator $XP_8(0|\mathbf{0}|22222222) = S^{\otimes 8}$ acts as a logical identity, but is not in the stabiliser group $\langle \mathbf{S}_X,\mathbf{S}_Z\rangle$. An interactive version of this example is in the \href{https://github.com/m-webster/CSSLO/blob/main/notebooks/01.0_logical_identity.ipynb}{linked Jupyter notebook}.
\end{example}

\subsection{Algorithm 2: Search for Diagonal XP Operator by Logical Action}\label{sec_ker_search}
We now demonstrate a method that searches for diagonal logical operators of XP form with a desired action. 
Aside from verifying if a CSS code has a transversal implementation of a particular logical operator, this is a useful method for cross-checking other algorithms.


\begin{framed}
\paragraph{Algorithm 2: Search for Diagonal XP Operator by Logical Action}
\paragraph{Input:}
 \begin{enumerate}
    \item The $X$-checks $S_X$ and $X$-logicals $L_X$ of a CSS code (\Cref{sec_css_codes});
     \item A level-$t$ controlled-phase operator $B$ on $k$ qubits (\Cref{sec_controlled_phase_operators}) such that $B\ket{\mathbf{0}} = \ket{\mathbf{0}}$ .
 \end{enumerate}
 \paragraph{Output:}
 A diagonal XP operator of precision $N = 2^t$ which acts as a logical $B$ operator or \texttt{FALSE} if this is not possible. 
 \paragraph{Method:}
\begin{enumerate}
    \item For $\mathbf{v} \in \mathbb{Z}_2^k$ calculate the phase $\mathbf{q_v} \in \mathbb{Z}_N$ such that $B\ket{\mathbf{v}} = \omega^{2\mathbf{q_v}}\ket{\mathbf{v}}$;
    \item Form the matrix $E_B$ that has rows of form $(-\mathbf{q_v}|\mathbf{e_{uv}})$ where $\mathbf{e_{uv}}:= \mathbf{u}S_X + \mathbf{v}L_X$;
    \item Calculate the kernel $K_B :=\ker_{\mathbb{Z}_N}(E_B)$;
    \item If there is an element $(1|\mathbf{z}) \in K_B$ then $\mathbf{z}$ is the $Z$-component of a logical $B$ operator $\overline{B}:= XP_N(0|\mathbf{0|z})$. This is because $(1|\mathbf{z}) \cdot (-\mathbf{q_v}|\mathbf{e_{uv}}) = 0\mod N \iff  \mathbf{e_{uv}} \cdot \mathbf{z} = \mathbf{q_v}\, \mod N$ for all $\mathbf{e_{uv}}$, which corresponds to the action of a logical $B$ operator on the codewords $\ket{\mathbf{v}}_L$.
\end{enumerate}
\end{framed}
The above algorithm requires that $B|\mathbf{0}\rangle =|\mathbf{0}\rangle$. If this is not the case, let $B|\mathbf{0}\rangle = \omega^p|\mathbf{0}\rangle$, run the algorithm using $B':=\omega^{-p}B$ and adjust for phase on the result. The results of the algorithm are dependent on the choice of basis for the span $\langle L_X\rangle$ (see \Cref{sec_css_codes}).

The logical action search algorithm involves finding the kernel of a matrix $E_B$ of dimension $2^{r+k} \times (n + 1)$. Hence the  complexity of the algorithm is sensitive to the number of logical qubits $k$ and independent $X$-checks $r$, but can be reduced as follows. 
Due to \Cref{prop_action_weight}, where $N = 2^t$ the dot product $\mathbf{e_{uv} \cdot z}$ can always be written as a $\mathbb{Z}_N$ linear combination of terms of form $\mathbf{e_{u'v'} \cdot z}$ where $\text{wt}(\mathbf{u'}) + \text{wt}(\mathbf{u'}) \le t$.
Hence, we only need to consider $\mathbf{e_{uv}}$ where $\text{wt}(\mathbf{u}) + \text{wt}(\mathbf{v}) \le t$ and $\mathbf{q_v}$ where $\text{wt}(\mathbf{v}) \le t$. The number of rows required in $E_M$ is therefore $\genfrac[]{0pt}{2}{k+r}{t}$ where $\genfrac[]{0pt}{2}{r}{t} := \sum_{0 \le j \le t}\genfrac(){0pt}{2}{r}{j}$. The dimensions of the matrix $E_M$ are polynomial rather than exponential in $n, k$ and $r$.
\begin{example}[Search for Diagonal XP Operator by Logical Action]\label{eg_ker_search}
The \href{https://github.com/m-webster/CSSLO/blob/main/notebooks/01.2_kernel_search.ipynb}{linked Jupyter notebook} illustrates the operation of the search algorithm on the hypercube code of \Cref{eg_hypercube_LI}. Users can enter the desired logical operator to search for in text form - for example \texttt{CZ[1,2]}, \texttt{S[1]} or \texttt{CCZ[0,1,2]}. 
The script either returns a diagonal XP operator with the desired logical action, or \texttt{FALSE} if there is no such operator. 
We find logical operators $\overline{CZ_{12}} = XP_8(0|\mathbf{0}|02060602)$ and $\overline{CCZ_{012}} = XP_8(0|\mathbf{0}|13313113)$ but no solutions for transversal logical $S$ operators.
\end{example}

\subsection{Logical Operator Test for Diagonal XP Operators}\label{sec_LO_test}
We now present an efficient method for determining whether a given diagonal XP operator acts as a logical operator on a CSS code, which relies on a commutator property of logical operators. This is used to find a generating set of all diagonal logical XP operators of a given precision and to check the results of other algorithms. 

Due to proposition E.2 of Ref.~\cite{XP}, an XP operator $B$ acts as a logical operator on the codespace if and only if the group commutator with any logical identity $A$ is again an element of the logical identity group $\mathcal{I}_\text{XP}$ (see \ref{sec_XP_groups}). That is:
\begin{align}
    [[A,B]] := ABA^{-1}B^{-1} \in \mathcal{I}_\text{XP}, \forall A \in \mathcal{I}_\text{XP}.
\end{align}
When $B := XP_N(0|\mathbf{0}|\mathbf{z})$ is diagonal and $A := XP_N(0|\mathbf{x}|\mathbf{0})$ is non-diagonal, by applying the COMM rule of Table 4 in Ref.~\cite{XP} we have:
\begin{align}
    [[A,B]] = XP_N(2\mathbf{x}\cdot\mathbf{z}|\mathbf{0}|-2\mathbf{x}\mathbf{z}).\label{eq_XP_commutator}
\end{align}
As  $B$ is a diagonal operator, we only need to consider commutators with non-diagonal elements of the logical identity group. In \Cref{prop_comm_X_checks} we show that this reduces to finding $\mathbf{z} \in \mathbb{Z}_N^n$ such that for all $X$-checks $\mathbf{x}$, both $\mathbf{x}\cdot \mathbf{z} = 0 \mod N$ and $2\mathbf{x}\mathbf{z} \in \langle K_M \rangle_{\mathbb{Z}_N}$  where $K_M$ is a generating set of $Z$-components of the diagonal logical identities as defined in \Cref{sec_LI} and $\langle K_M \rangle_{\mathbb{Z}_N}$ the row span of $K_M$ over $\mathbb{Z}_N$.

As $2\mathbf{xz}$ and $N$ are both divisible by 2, we apply the method of \Cref{sec_XP} and see that the group commutator must be at most a level $t-1$ Clifford hierarchy operator. 
For instance, for $t=2, N=4$ logical operators must commute up to level $t=1, N=2$ logical identities which are the $Z$-checks (see \Cref{eg_hypercube_LI}).
This observation either eliminates the need to calculate the logical identities (for $t \le 2$) or reduces the complexity of calculating them (the number of rows in the matrix $E_M$ of \Cref{sec_LI} is a polynomial of degree $t$).

Hence, we have an $\mathcal{O}(r)$ algorithm for checking whether a diagonal XP operator is a logical operator of a CSS code where $r$ is the number of independent $X$-checks, but we may need to first run the diagonal logical identity algorithm of \Cref{sec_LI} at level $t-1$. 
\begin{framed}
\paragraph{Algorithm 3: Logical Operator Test for Diagonal XP Operators}
\paragraph{Input:}
\begin{enumerate}
    \item The $X$-checks $S_X$ of a CSS code (\Cref{sec_css_codes});
    \item The matrix $K_M$ corresponding to the $Z$-components of the level $t-1$ diagonal logical identity generators (\Cref{sec_LI};
    \item A diagonal XP operator $B = XP_N(0|\mathbf{0|z})$ on $n$ qubits of precision $N = 2^t$ (\Cref{sec_XP}).
\end{enumerate}
\paragraph{Output:}
\texttt{TRUE} if $B$ acts as a logical operator on the code or \texttt{FALSE} otherwise.
\paragraph{Method:}
\begin{enumerate}
    \item For each row $\mathbf{x}$ of $S_X$:
    \begin{enumerate}
        \item Check if $\mathbf{x\cdot z} = 0 \mod N$; and
        \item Check if $2\mathbf{xz}$ is in the rowspan of $K_M$ over $\mathbb{Z}_N$; 
        \item If either is not the case, return \texttt{FALSE}.
    \end{enumerate} 
    \item Return \texttt{TRUE}.
\end{enumerate}
\end{framed}

\begin{example}[Logical Operator Test]
In this example, we apply the logical operator test to the logical $\overline{CZ_{12}}$ found for the hypercube code in \Cref{eg_ker_search}. As $\overline{CZ_{12}}:=XP_8(0|\mathbf{0}|02060602)$, we let $\mathbf{z} = 02060602$. Let $\mathbf{x} = 11111111$ corresponding to the single $X$-check. We calculate the group commutator $C:= (2\mathbf{x}\cdot \mathbf{z}|\mathbf{0}|-2\mathbf{x} \mathbf{z})$. We find that $\mathbf{x}\cdot\mathbf{z} = 16 = 0 \mod 8$ and $-2\mathbf{xz} = 04040404 \mod 8$. Referring to \Cref{eg_hypercube_LI}, we see that this vector is a row of $K_M$. As both $\mathbf{x}\cdot\mathbf{z} = 0 \mod 8$ and $-2\mathbf{x} \mathbf{z} \in \langle K_M\rangle_{\mathbb{Z}_N}$, $C$ is a logical identity. Accordingly, we have verified that $\overline{CZ_{12}}$ is a diagonal logical operator on the code. Applying the method of \Cref{sec_XP}, we note that $\overline{CZ_{12}}$ is at level $2$ of the Clifford hierarchy and the group commutator $C$ is at level $1$.
\end{example}

\subsection{Diagonal Logical XP Operator Group Generators}\label{sec_comm_LO}
We now show how to apply the test for diagonal logical XP operators of \Cref{sec_LO_test} to find all diagonal logical operators of XP form for a CSS code. 
\begin{framed}
\paragraph{Algorithm 4: Diagonal Logical XP Operator Group Generators}
\paragraph{Input:}
\begin{enumerate}
    \item The $X$-checks $S_X$ of a CSS code (\Cref{sec_css_codes});
    \item The desired level of the Clifford hierarchy $t$ (\Cref{sec_controlled_phase_operators});
    \item The matrix $K_M$ corresponding to the $Z$-components of the level $t-1$ diagonal logical identity generators (\Cref{sec_LI}).
\end{enumerate}
\paragraph{Output:} A matrix $K_L$ over $\mathbb{Z}_N$ representing the $Z$-components of a generating set of diagonal logical operators of XP form (\Cref{sec_XP_groups}). 
\paragraph{Method:}
\begin{enumerate}
    \item For each $X$-check $\mathbf{x} \in S_X$, find solutions $\mathbf{z} \in \mathbb{Z}_N^n$ such that both $\mathbf{x}\cdot \mathbf{z} = 0$ and $2\mathbf{x} \mathbf{z} \in \langle K_M\rangle_{\mathbb{Z}_N}$. Details of solving within these constraints are set out in \Cref{sec_comm_x}. Denote the solutions $\text{Comm}_N(K_M, \mathbf{x})$;
    \item Find the intersection of all such solution sets $K_L := \bigcap_{\mathbf{x} \in S_X}\text{Comm}_N(K_M, \mathbf{x})$. The method for determining intersections of spans over $\mathbb{Z}_N$ is covered in Appendix A.4 of \cite{XP};
    \item Return $K_L$.
\end{enumerate}
\end{framed}
The rows of $K_L$ correspond to the $Z$-components of a generating set of the logical XP group (\Cref{sec_XP_groups}), which includes the logical identity XP group. Determining the logical action of the operators is discussed in \Cref{sec_logical_action}.
\subsection{Determine Action of Diagonal Logical XP Operator}\label{sec_logical_action}
Here we demonstrate an algorithm expressing the action of a diagonal logical XP operator in terms of logical controlled-phase operators. This is important because the algorithm in \Cref{sec_comm_LO} does not yield any information on the action of the resulting diagonal logical operators. 
\begin{framed}
\paragraph{Algorithm 5: Determine Action of Diagonal Logical XP Operator}
\paragraph{Input:}
\begin{enumerate}
    \item The $X$-logicals $L_X$ of a CSS code (\Cref{sec_css_codes}) with $k$ logical qubits;
    \item A diagonal XP operator $\overline{B}$ of precision $N:= 2^t$ that acts as a logical operator on the code (\Cref{sec_LO_of_CSS_codes}).
\end{enumerate}
\paragraph{Output:}
A diagonal Clifford hierarchy operator $B$ on $k$ qubits representing the logical action of $\overline{B}$.
\paragraph{Method:}
\begin{enumerate}
    \item Let $V:=\{\mathbf{v} \in \mathbb{Z}_2^n : \text{wt}(\mathbf{v}) \le t\}$; 
    \item For each $\mathbf{v} \in V$, calculate $\mathbf{q_v}$ such that  $\overline{B}\ket{\mathbf{v}L_X} = \omega^{\mathbf{q_v}}\ket{\mathbf{v}L_X}$;
    \item Loop over each $\mathbf{v} \in V$ ordered by weight. For any $\mathbf{v} \preccurlyeq \mathbf{u} \in V \setminus \{\mathbf{v}\}$, update $\mathbf{q_u} := (\mathbf{q_u} - \mathbf{q_v})\mod 2N$;
    \item Return $B :=\prod_{\mathbf{v} \in V}CP_N(\mathbf{q_v}, \mathbf{v})$ in terms of the vector form of controlled-phase operators of \Cref{sec_CP_vector}.
\end{enumerate}
\end{framed}
The above algorithm involves calculating $\mathcal{O}(k^t)$ phase components $\mathbf{q_v}$, and this is sufficient due to \Cref{prop_action_weight}. A naive approach which calculates the phase applied to each codeword would involve calculating $\mathcal{O}(2^k)$ such phase components, and would be impractical for CSS codes with a large number of logical qubits. The results of the algorithm are dependent on the choice of basis for the span $\langle L_X\rangle$ (see \Cref{sec_css_codes}).

\begin{example}[Action of Diagonal Logical XP Operators - Hypercube Codes]
In this example, we apply the method of \Cref{sec_comm_LO} to the Hypercube code of \Cref{eg_hypercube_LI}  at level $t=3$. The output of the method of \Cref{sec_comm_LO} is a set of length $8$ vectors over $\mathbb{Z}_8$ corresponding to $Z$-components of diagonal logical XP operators. Using the method of \Cref{sec_logical_action}, we obtain the following list of logical actions corresponding to the $Z$-components: 
$$
\begin{array}
{|l|l|l|}
\hline
\mathbf{z} &\textbf{Logical Action}& \textbf{Clifford level}\\
\hline
00000044  &Z_{0}&1\\
00000404  &Z_{1}&1\\
00040004  &Z_{2}&1\\
00002662 &CZ_{01}&2\\
00260062  &CZ_{02}&2\\
02060602   &CZ_{12}&2\\
13313113   &CCZ_{012}&3\\
\hline
\end{array}
$$
In \Cref{fig_hypercube} we display the resulting logical operators on the cube and notice that that the Clifford hierarchy level of the logical operator corresponds to the dimension of the support of the operator. 
An interactive version of this example is available in the \href{https://github.com/m-webster/CSSLO/blob/main/notebooks/02.1_commutator_small.ipynb}{linked Jupyter notebook}.

\begin{figure}[hbt!]
\centering
\includegraphics[width=0.9\textwidth]{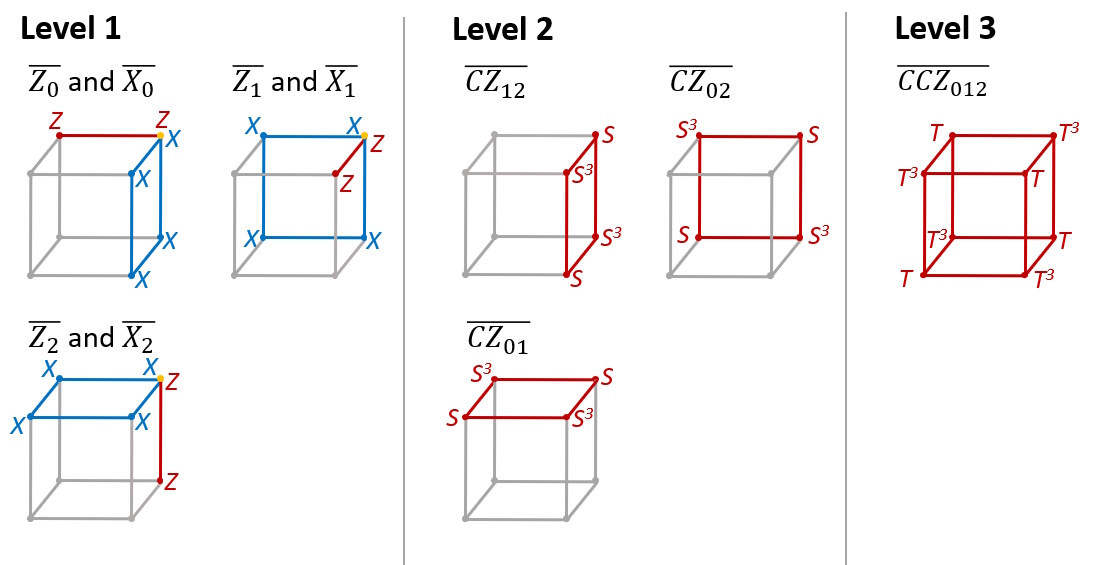}

\caption{\textbf{Diagonal Logical Operators of Hypercube Code of Dimension $3$}:  Logical XP operators returned by the method of \Cref{sec_comm_LO} are plotted on the cube. 
We note that Clifford hierarchy level 1 logical Z operators have support on 1D edges, level 2 CZ operators have support on 2D faces and level 3 CCZ operators have support on the entire 3D cube.}\label{fig_hypercube}
\end{figure}
\end{example}

\begin{example}[Hyperbolic Quantum Color Codes and Poset Codes]\label{eg_comm_large}
In the \href{https://github.com/m-webster/CSSLO/blob/main/notebooks/02.2_commutator_large.ipynb}{linked Jupyter notebook}, we illustrate the application of the method of \Cref{sec_comm_LO} to codes that have a large number of logical qubits. We choose examples of self-orthogonal codes that are known to have transversal implementations of diagonal level 2 Clifford hierarchy logical operators.

Hyperbolic quantum color codes~\cite{LDPC_colour} involve constructing codes from tesselations of the 2D hyperbolic plane. The tesselations are formed from polygons with an even number of edges, and each vertex is shared by 3 such polygons. We place a qubit on each vertex of the tesselation. For each polygonal face, we have an $X$-check corresponding to the adjacent vertices. The $Z$-checks are the same as the $X$-checks. Applying the method of \Cref{sec_comm_LO}, we find that the codes have a transversal level $2$ logical operator with action which can be expressed as a product of controlled-Z operators.

\captionsetup[subfigure]{margin=5pt}
\begin{figure}[hbt!]
\centering
\begin{subfigure}[t]{.5\textwidth}
  \centering
  \includegraphics[width=0.9\linewidth]{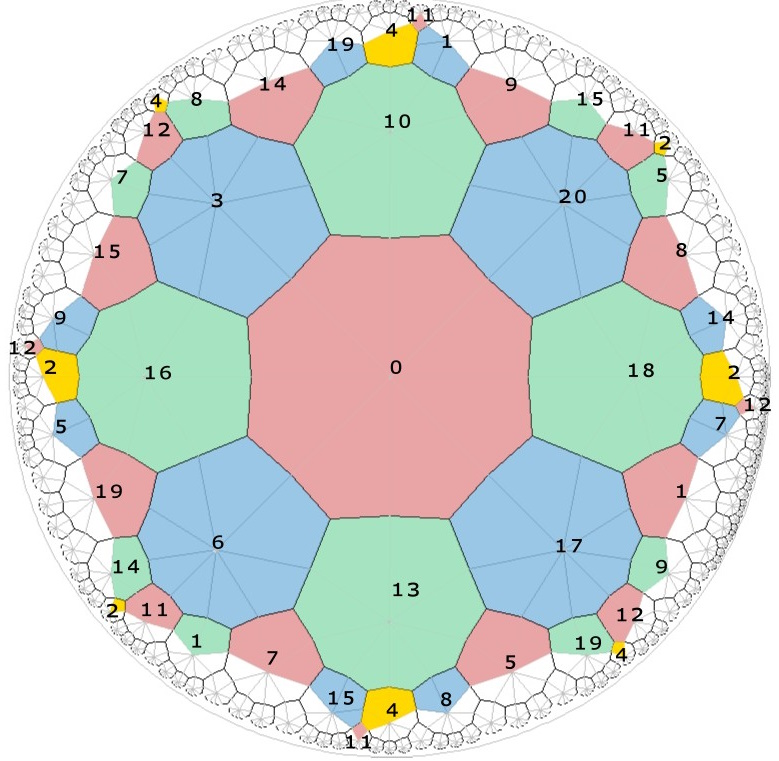}
    \subcaption{{[[56,14,6]] Code}  }
    \end{subfigure}%
\begin{subfigure}[t]{.5\textwidth}
  \centering
      \includegraphics[width=0.9\linewidth]{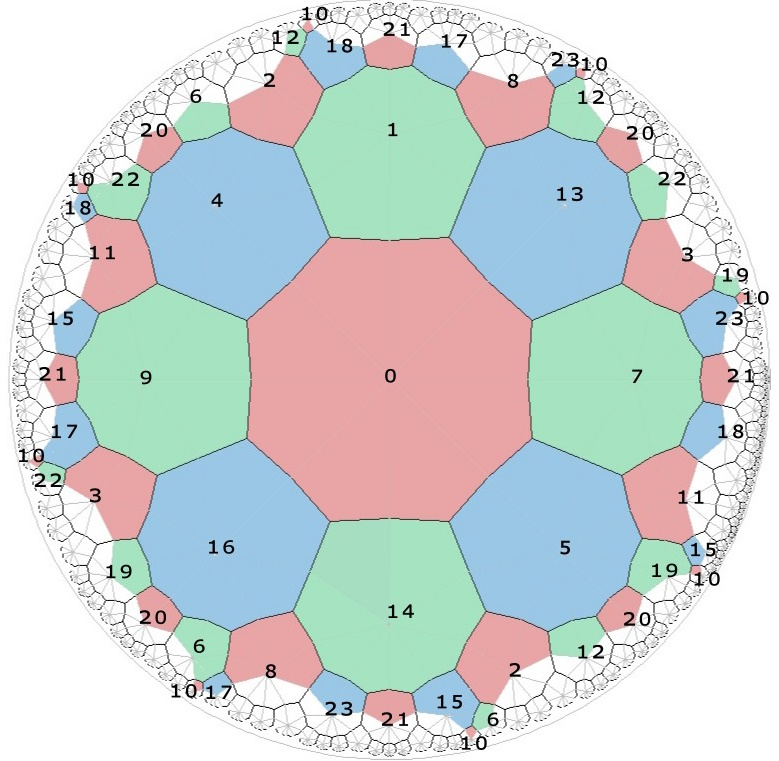}
    \subcaption{{[[64,20,4]] Code} }
     \end{subfigure}
     \caption{\{8,3\} Hyperbolic Color Codes: above are tesselations corresponding to two hyperbolic color codes from \Cref{eg_comm_large}. The [[56,14,6]] code is not globally 3-colourable as there is no valid colour assignment for faces 2 and 4 in the diagram above. Each code has a transversal level 2 diagonal logical operator whose action is a product of logical CZ operators. }
     \label{fig_hyperbolic_colour}
\end{figure}

There are various methods in the literature for constructing classical self-orthogonal codes and these can also be used to make quantum codes with $Z$-checks which are the same as the $X$-checks which we expect to have transversal level 2 diagonal logical operators.  In ~\cite{poset_codes}, self-orthogonal codes are constructed from partially ordered sets (posets). Analysing poset codes using our methods, we see that they have transversal level $2$ logical operators with actions which can be expressed as products of $S$ and $CZ$ operators.
\end{example}

\begin{example}[Triorthogonal Codes]

For triorthogonal codes \cite{triorthogonal}, there is always a logical operator of form $\overline{T^{\otimes k}} := UT^{\otimes n}$ where $U$ is a product of $CZ$ and $S$ operators and $k$ is the number of logical qubits of the code. In the \href{https://github.com/m-webster/CSSLO/blob/main/notebooks/02.3_commutator_triorthogonal.ipynb}{linked Jupyter notebook}, we apply the method of \Cref{sec_comm_LO} to find generating sets of diagonal logical operators for the 38 triorthogonal code classes in Table II of Ref.~\cite{small_triorthogonal}. In this example, we consider codes with $k=3$ logical qubits (this choice can be modified by the user).  

Applying our method, we see that the logical operator structure of triorthogonal codes varies widely. In some cases, the code has a transversal logical $\overline{T_i}$ operator for each logical qubit $0 \le i < k$. For most of the codes, we find a logical $\overline{T^{\otimes k}}$ operator of XP form. The exceptions are codes which require the application of $CZ$ operators to form $\overline{T^{\otimes k}}$, and so would not be identified by our method. We do not see any instances of logical CCZ or CS operators. 
\end{example}

\section{Transversal Logical Operators Composed of Multi-Qubit Controlled-Phase Gates}
In the previous Section, we have shown how to find a generating set of all logical operators of a CSS code that can be constructed from single-qubit phase gates at any level of the Clifford hierarchy. This relied on representing operators composed of single-qubit phase gates as diagonal XP operators. In this Section, we show how to find all transversal (depth-one) logical operators of a CSS code composed of multi-qubit controlled-phase gates. The method relies on representing controlled-phase operators acting on a codespace as diagonal XP operators acting on a larger Hilbert space via an embedding operator. 

The structure of this Section is as follows. We first introduce phase-rotation gates and discuss some of their elementary properties. We then prove a duality result that transforms controlled-phase operators to phase-rotation operators and vice versa. Hence phase-rotation gates are an alternative generating set for diagonal Clifford hierarchy gates. We then describe an embedding operator from the codespace into a larger Hilbert space such that  phase-rotation operators in the codespace  correspond to diagonal XP operators in the embedded codespace. As a result, any diagonal Clifford hierarchy operator can be represented as a diagonal XP operator in the embedded codespace. 

Finally, we demonstrate an algorithm that searches for transversal logical operators composed of single- and multi-qubit controlled-phase gates for a given CSS code. Such implementations are depth one and use operators with bounded support size and so have fault-tolerant properties. Logical operators of this type have recently been studied in Refs.~\cite{fold_transversal}, \cite{partitioning_qubits} and \cite{morphing} and we provide examples of the application of the algorithm to codes in these papers.

\subsection{Phase-Rotation Operators}\label{sec_phase_rotation}
Phase-rotation operators are single or multi-qubit diagonal gates that form an alternative generating set for the diagonal Clifford hierarchy operators of \Cref{sec_controlled_phase_operators}.   Phase-rotation operators are defined as follows. Let $A := XP_2(0|\mathbf{0}|\mathbf{v})$ be a tensor product of $Z$ operators and let $\omega := \exp(\pi i / N)$.  Let  $A_{\pm 1} : =(I \pm A)/2$ be the projectors onto the $\pm 1$ eigenspaces of $A$ and let $q \in \mathbb{Z}_{2N}$. The phase-rotation operator is:
\begin{align}
    \text{RP}_N(q,\mathbf{v}) = \exp(\frac{q\pi i}{N} A_{-1}).\label{eq_RP_def}
\end{align}
This form is similar to the Pauli product rotations of Ref.~\cite{litinski} and operators of this type arise as fundamental gates in NMR \cite{NMR} and quantum dot systems \cite{quantum_dots}. In  \Cref{prop_RP_action}, we show that the action of $\text{RP}_N(q,\mathbf{v})$ on the  computational basis element $|\mathbf{e}\rangle$ for $\mathbf{e} \in \mathbb{Z}_2^n$ is:
\begin{align}
    \text{RP}_N(q,\mathbf{v})|\mathbf{e}\rangle = \begin{cases}\omega^q|\mathbf{e}\rangle& \text{ if }\mathbf{e}\cdot\mathbf{v} \mod 2  = 1\\|\mathbf{e}\rangle& \text{ otherwise}.\end{cases}
\end{align}
We can express the phase applied more concisely as follows:
\begin{align}
    \text{RP}_N(q,\mathbf{v})|\mathbf{e}\rangle = \omega^{q \cdot s_\mathbf{v}(\mathbf{e})}|\mathbf{e}\rangle\text{ where }s_\mathbf{v}(\mathbf{e}) := \bigoplus_{i \preccurlyeq \mathbf{v}}\mathbf{e}[i].\label{eq_RP_phase}
\end{align}
Single qubit phase gates of precision $N$ in this notation are of form $P_i = \text{RP}_N(2, \mathbf{b}^n_i)$ where $\mathbf{b}^n_i$ is the length $n$ binary vector which is all zero, apart from component $i$ which is one.

Where the precision and number of qubits are fixed, we use a more concise notation for phase-rotation operators analogous to the notation for controlled-phase operators. For example, on $n=3$ qubits, the following are examples of precision $N=8$ operators: $RRZ_{012} := RP_8(8, 111), RS_{01} := RP_8(4,110), T_0 = RP_8(2,100)$.

\subsection{Duality of Controlled-Phase and Phase-Rotation Operators}\label{sec_duality}
In \Cref{prop_duality}, we prove a duality result that allows us to convert vector form controlled-phase operators to products of  phase-rotation operators and vice versa:
\begin{align}
    \text{CP}_N(2^{\text{wt}(\mathbf{v})},\mathbf{v}) = \prod_{\mathbf{0} \ne \mathbf{u} \preccurlyeq \mathbf{v}}\text{RP}_N(2 \cdot (-1)^{\text{wt}(\mathbf{u})-1}, \mathbf{u});\\
    \text{RP}_N(2,\mathbf{v}) = \prod_{\mathbf{0} \ne \mathbf{u} \preccurlyeq \mathbf{v}}\text{CP}_N(2 \cdot(-2)^{\text{wt}(\mathbf{u})-1}  , \mathbf{u}).
\end{align}
In \Cref{sec_CP_vector}, we saw that operators of form $\text{CP}_N(2^{\text{wt}(\mathbf{v})},\mathbf{v})$ with $\text{wt}(\mathbf{v}) \le t$ and $N := 2^t$  generate  the level $t$ diagonal Clifford hierarchy operators. As a consequence of the duality result,  phase-rotation operators of form $\text{RP}_N(2,\mathbf{v})$ where $\text{wt}(\mathbf{v}) \le t$ are an alternative generating set. In the \href{https://github.com/m-webster/CSSLO/blob/main/notebooks/10.1_CP_RP_duality.ipynb}{linked Jupyter notebook} we show that $RS_{01} = CZ_{01} S_1S_2 = RS^3_{01}Z_1 Z_2$ by applying the duality result twice - hence phase-rotation operators may have more than one vector representation.

\subsection{Embedded Code Method}\label{sec_embedded_codes}
The embedded code method involves constructing an embedding operator on the codespace of a CSS code such that phase-rotation operators in the original codespace correspond to diagonal XP operators in the embedded codespace. The embedding technique is similar to the one used to represent weighted hypergraph states in Section 5.4 of Ref.~\cite{XP}. We first define the embedding operator in terms of its action on computational basis states, then show how to extend it to phase rotation operators and strings of Pauli X operators. As an example, we show how the embedding operator transforms repetition codes.
\subsubsection{Action of Embedding Operator on Computational Basis States and CSS Codespaces}
Let $M^n_t$ be the matrix whose rows are the binary vectors of length $n$ of weight between $1$ and $t$. Let $V$ be a matrix whose rows are a subset of the rows of $M^n_t$. We define the embedding operator $\mathcal{E}_V: \mathcal{H}_2^n \rightarrow \mathcal{H}_2^{|V|}$ that has the following action on computational basis vectors $\ket{\mathbf{e}}, \mathbf{e} \in \mathbb{Z}_2^n$:
\begin{align}
\mathcal{E}_V|\mathbf{e}\rangle = |\mathbf{e}V^T \mod 2\rangle.\label{eq_embedding}
\end{align}
Now let $S_X, L_X$ be the $X$-checks and $X$-logicals of a CSS code $\mathbf{C}$ on $n$ qubits (see \Cref{sec_css_codes}).  The image of the codespace of $\mathbf{C}$ under $\mathcal{E}_V$ is the codespace of the embedded code $\mathbf{C}_V$ defined as follows:
\begin{itemize}
    \item $X$-checks $S^V_X:=  S_X V^T$
    \item $X$-logicals $L^V_X := L_X V^T$
    \item $Z$-checks $S^V_Z := \ker_{\mathbb{Z}_2}\begin{pmatrix}S^V_X\\L^V_X\end{pmatrix}$
\end{itemize}
Providing $V$ is full rank, the $X$-checks and $X$-logicals of the embedded code are independent (for instance if $V$ includes all rows of $I_n$). We will show that phase-rotation operators acting on the codespace correspond to diagonal XP operators in the embedded codespace. Because operators of form $RP_N(2,\mathbf{v})$ for $\mathbf{v} \in M^n_t$ and $N=2^t$ generate all controlled-phase operators of level $t$ on $n$ qubits (see \Cref{sec_duality}), choosing $V=M^n_t$ allows any such operator to be represented. By limiting $V$ to a subset of $M^n_t$, we can place restrictions on the phase-rotation operators we wish to work with in the embedded codespace. For instance, we can allow only nearest neighbour interactions for a lattice-based code or cater for ZX symmetries and qubit partitions as discussed in Refs.~\cite{fold_transversal} and \cite{partitioning_qubits}.

\subsubsection{Action of Embedding Operator on Phase-Rotation and Pauli X Operators}\label{sec_embedding_operators}
We now demonstrate an extension of the embedding operator $\mathcal{E}_V$ to phase-rotation and Pauli X operators which acts as a group homomorphism. A group homomorphism must respect commutation relations, and this is much simpler to achieve for phase-rotation operators than for controlled-phase operators. In \Cref{prop_CP_relation}, we prove the following commutation relation for controlled-phase and Pauli X operators:
\begin{align}
\text{CP}_N(q,\mathbf{v}) X_i  = \begin{cases}X_i\text{CP}_N(-q,\mathbf{v})\text{CP}_N(q,\mathbf{v}\oplus \mathbf{b}^n_i)& \text{ if } \mathbf{v}[i] = 1\\X_i\text{CP}_N(q, \mathbf{v})& \text{ otherwise.}\end{cases}\label{eq_cpXi}
\end{align}
In \Cref{eq_cpXi}, $\mathbf{b}^n_i$ is the binary vector of length $n$ which is zero apart from entry $i$ which is one. Extending this to arbitrary strings of $X$ operators we obtain the following:
\begin{align}
\text{CP}_N(q, \mathbf{v}) XP_2(0|\mathbf{x}|\mathbf{0})  &= XP_2(0|\mathbf{x}|\mathbf{0}) \prod_{\mathbf{0} \preccurlyeq \mathbf{u} \preccurlyeq \mathbf{xv}}\text{CP}_N(q\cdot (-1)^{\text{wt}(\mathbf{xv})+\text{wt}(\mathbf{u})},\mathbf{v}\oplus\mathbf{u}).\label{eq_CP_relation}
\end{align}
In \Cref{prop_RP_algebra}, we prove the much simpler commutation relation for phase-rotation operators which corresponds closely to the commutation relation for XP operators in \Cref{eq_XP_relation}:
\begin{align}
    \text{RP}_N(q,\mathbf{v})X_i = \begin{cases} \omega^qX_i\text{RP}_N(-q,\mathbf{v})  & \text{ if } \mathbf{v}[i] = 1\\ X_i \text{RP}_N(q,\mathbf{v}) & \text{ otherwise}\end{cases}\label{eq_relation_RP}
\end{align}
The relation in \Cref{eq_relation_RP} also implies that for any $V \subset M^n_t$, we have closure under conjugation with any Pauli $X$ string, which is not the case for controlled-phase operators. 

Now consider the group $\mathcal{XRP}^V_N$ generated by operators of form $\omega I$, $X_i$ and $\text{RP}_N(2, \mathbf{v})$ for $\mathbf{v}$ a row of $V$. Elements of $\mathcal{XRP}^V_N$ can be written in terms of components $p \in \mathbb{Z}_{2N}$, $\mathbf{x} \in \mathbb{Z}_2^n$, $\mathbf{q} \in \mathbb{Z}_N^{|V|}$ such that:
\begin{align}
    \text{XRP}^V_N(p|\mathbf{x}|\mathbf{q}):= \omega^p\prod_{0 \le i < n}X_i^{\mathbf{x}[i]}\prod_{\mathbf{v} \in V}\text{RP}_N(2\mathbf{q}[\mathbf{v}],\mathbf{v}).
\end{align}
We define an embedding map for XRP operators with respect to $V$ as follows:
\begin{align}
    \mathcal{E}_V(\text{XRP}^V_N(p|\mathbf{x}|\mathbf{q})) := XP_N(p|\mathbf{x}V^T|\mathbf{q})
\end{align}
In \Cref{prop_emb_hom}, we show that the embedding operator $\mathcal{E}_V$ respects group operations and so acts as a group homomorphism. As a result, we can use the diagonal logical identity and logical operator algorithm in \Cref{sec_LI} and \Cref{sec_comm_LO} to find logical operators in the embedded codespace. The results can be interpreted as phase-rotation operators in the original codespace. One application of this method is to better understand what kinds of coherent noise a CSS code is inherently protected against as in \cite{coherent_noise}. The logical identity group of the embedded code represents the correlated noise that the code protects against up to configurable constraints (for example connectivity and the level of Clifford hierarchy).

\begin{example}[Embedding the Repetition Code]\label{eg_embedding_rep}
In this example we show how to construct an embedded code based on the repetition code. For example, let $S_X$ be the check matrix of the classical repetition code on $3$ bits and let $L_X$ be a weight one vector. This forms a CSS code $\mathbf{C}$ with:
\begin{align}
    S_X &:= \begin{pmatrix}110\\011\end{pmatrix},\\
    L_X &:= \begin{pmatrix}001\end{pmatrix}.
\end{align}
Let $V := M^3_2$ be the matrix whose rows are binary vectors of length $3$ and weight $1$ or weight $2$. The embedded code $\mathbf{C}_V$ is defined by setting $S^V_X := S_X V^T$ and $L^V_X := L_X V^T$ so that:
\begin{align}
    V^T &:= \begin{pmatrix}100110\\010101\\001011\end{pmatrix},\\
    S^V_X &:= S_X V^T = \begin{pmatrix}101101\\011110\end{pmatrix},\\
    L^V_X &:= L_X V^T =  \begin{pmatrix}001011\end{pmatrix}.
\end{align}
Applying the method of \Cref{sec_comm_LO}, we find that the embedded code has a logical $S$ operator given by $\overline{S}_V:=XP_4(0|\mathbf{0}|113133) = S_0S_1S_2^3S_3S_4^3S_5^3$. In the original codespace, this corresponds to the following product of phase-rotation gates (\Cref{sec_phase_rotation}): 
\begin{align}\overline{S} := RP_4(2, 100)RP_4(2, 010)RP_4(6, 001) RP_4(2, 110) RP_4(6, 101)RP_4(6, 011).\end{align}

In the \href{https://github.com/m-webster/CSSLO/blob/main/notebooks/03.1_embedded_repetition.ipynb}{linked Jupyter notebook}, users can verify that using a repetition code on $d$ bits and $V = M^d_t$ the matrix whose rows are binary vectors of length $d$ of weight between $1$ and $t$, the embedded code has a transversal logical phase gate at level $t$ of the Clifford hierarchy.
\end{example}

\subsection{Algorithm 6: Depth-One Logical Operators}\label{sec_depth_one_algorithm}
We now show how to find the transversal logical operators composed of single and multi-qubit diagonal Clifford hierarchy gates (i.e., depth-one circuit implementations where each physical qubit is involved in at most one gate) for a CSS code. It relies on the method of representing phase-rotation operators on a codespace as XP operators in an embedded codespace of \Cref{sec_embedded_codes}.

\begin{framed}
\paragraph{Algorithm 6: Depth-One Logical Operators}
\paragraph{Input:}
\begin{enumerate}
    \item  The $X$-checks $S_X$ and $X$-logicals $L_X$ of a CSS code (\Cref{sec_css_codes});
    \item The desired level $t$ of the Clifford hierarchy (\Cref{sec_controlled_phase_operators}).
\end{enumerate}
\paragraph{Output:}
 A depth-one implementation of a logical controlled-phase operator at level $t$, or \texttt{FALSE} if there is no such implementation.
\paragraph{Method:}
\begin{enumerate}
    \item Use the embedding $V = M^n_t$ -- all binary vectors of length $n$ of weight between $1$ and $t$;
    \item For the embedded code $\mathbf{C}_V$ (\Cref{sec_embedded_codes}), calculate $K_L$ the rows of which are the $Z$-components of a generating set of the diagonal logical XP operator group (\Cref{sec_comm_LO});
    \item For each row of $K_L$, determine the logical action and the level of the Clifford hierarchy (\Cref{sec_logical_action});
    \item From the rows of $K_L$, choose a vector $\mathbf{z}$ corresponding to a logical operator at level $t$ of the Clifford hierarchy. If there is no such operator, return FALSE. Otherwise, perform the following steps:
    \begin{enumerate}
    \item Remove $\mathbf{z}$ from $K_L$;
    \item For each element $\mathbf{q}$ of the rowspan of $K_L$ over $\mathbb{Z}_N$, check if $\mathbf{z}' := (\mathbf{q} + \mathbf{z}) \mod N$ represents a depth-one operator. If so, return $\mathbf{z}'$;
    \item If no depth-one operator is found, go to Step 4.
    \end{enumerate}
\end{enumerate}
\end{framed}
When the CSS code has a known symmetry, we can search for depth-one logical operators more efficiently by modifying the embedding operator. The depth-one algorithm can take as input a permutation of the physical qubits in cycle form such that the cycles partition the $n$ physical qubits.  Let $\mathbf{c} = (c_1,c_2,\dots,c_l)$ be a cycle in the permutation and let $\mathbf{b}^n_\mathbf{c}$ be the length $n$ binary vector which is zero apart from the components $i \in \mathbf{c}$ that are one. The rows of the embedding matrix $V$ are the vectors $\mathbf{b}^n_\mathbf{c}$ for the cycles $c$ in the permutation.

The algorithm as outlined above yields logical operators composed of physical phase-rotation gates. To search for logical operators composed of controlled-phase gates, transform the matrix  $K_L$ by using the duality result of \Cref{sec_duality}. In this case, due to the commutation relation in \Cref{eq_CP_relation}, we need to ensure that for all $\mathbf{v} \in V$ any length $n$ binary vector whose support is a subset of the support of $\mathbf{v}$ is also in $V$ -- that is  if $\mathbf{v} \in V$ and $\mathbf{u} \preccurlyeq \mathbf{v}$ then $\mathbf{u} \in V$.

Note that $M^n_t$ has $\genfrac[]{0pt}{2}{n}{t}:=\sum_{1 \le j \le t}\genfrac(){0pt}{2}{n}{j}$ rows, so we would generally apply the algorithm only to small codes of around 30 physical qubits (e.g., for $t=2$ and $n=30, M^n_t$ has $465$ rows, but for $n=100, M^n_t$ has $5,050$ rows).  In \Cref{app_depth_one_algorithm} we describe a method for more efficiently exploring the search space. 

\begin{example}[Depth-One Algorithm]
In the \href{https://github.com/m-webster/CSSLO/blob/main/notebooks/03.2_depth_one_search.ipynb}{linked Jupyter notebook}, we illustrate the depth-one search algorithm for small codes. For a given code and a desired level of the Clifford hierarchy $t$, the output is a logical operator with a depth-one circuit implementation whose logical action is at level $t$ of the diagonal Clifford hierarchy, or FALSE if no such operator exists. This is done with no knowledge of the logical action of the operator or symmetries of the code.  For example, we identify the depth-one implementation of the logical $\overline{S}_0\overline{S}_1^3$ of the 2D toric code as discussed in Refs.~\cite{moussa}, \cite{fold_transversal} and \cite{partitioning_qubits}. Users can also apply the algorithm to Bring's code which is a 30-qubit LDPC code discussed in Ref.~\cite{fold_transversal} and various examples of morphed codes which are discussed in Ref.~\cite{morphing}. Users can also choose to use a known symmetry of the code to speed up the search -- this can be used for instance to verify the partitioned logical operators of the symmetric hypergraph product codes of Ref.~\cite{partitioning_qubits}.
\end{example}

\section{Other Applications of Embedded Codes}
In this Section, we discuss other applications of the embedded code method of \Cref{sec_embedded_codes}. We first show that for any CSS code with $k$ logical qubits and any diagonal Clifford hierarchy operator $B$ on $k$ qubits, we can write a closed-form expression for a logical $\overline{B}$ operator on the codespace composed of phase-rotation gates (see \Cref{sec_phase_rotation}). As a consequence, the embedded code has a logical $B$ operator composed of single-qubit phase gates. This leads to a method of generating CSS codes that have transversal implementations of any desired diagonal logical Clifford hierarchy operator.

\subsection{Canonical Implementations of Logical Controlled-Phase Operators}\label{sec_canonical_logical_operators}
Here, we show how to implement a desired logical controlled-phase operator on an arbitrary CSS code  via a canonical form composed of the phase-rotation gates of \Cref{sec_phase_rotation}.  
We demonstrate implementations of logical $S, T, CZ$ and $CS$ operators using the 2D toric code as an example. As the canonical implementation is in terms of phase-rotation operators, we can apply the embedded code method of \Cref{sec_embedded_codes} and implement the logical operator in the embedded codespace using single qubit phase gates. We use this fact to generate families of CSS codes that have transversal implementations of a desired logical controlled-phase operator using single-qubit phase gates. The methodology is illustrated in \Cref{fig_embedded_codes}. 

\begin{figure}[hbt!]
\centering
$$
\begin{CD}
\mathcal{H}_2^k @>\mathcal{C}>> \mathcal{H}_2^n @>\mathcal{E}_V>>  \mathcal{H}_2^{|V|}\\
@VBVV @V\overline{B}VV @V\overline{B}_{V}VV\\
\mathcal{H}_2^k @>\mathcal{C}>> \mathcal{H}_2^n @>\mathcal{E}_V>>  \mathcal{H}_2^{|V|}
\end{CD}
$$
\caption{\textbf{Logical Operators of CSS Codes and Embedded Codes}:  A CSS encoding maps $k$ logical qubits into $n$ physical qubits via $\mathcal{C}: \mathcal{H}_2^k \rightarrow \mathcal{H}_2^n$, which takes computational basis elements $|\mathbf{v}\rangle$ to codewords $|\mathbf{v}\rangle_L$. Consider a  level-$t$ controlled-phase operator $B$ acting on $\mathcal{H}_2^k$. An operator $\overline{B}$ acting on $\mathcal{H}_2^n$ is a logical $B$ operator if $\overline{B}\mathcal{C} = \mathcal{C}B$. We show how to construct a canonical logical $B$ operator $\overline{B}$ from level-$t$  phase-rotation gates. Let $V$ be the matrix whose rows are length $n$ binary vectors representing the support of the controlled-phase operators making up $\overline{B}$. The embedded codespace is formed by applying the embedding $\mathcal{E}_V:\mathcal{H}_2^n \rightarrow \mathcal{H}_2^{|V|}$ which takes the computational basis element $\ket{\mathbf{e}}$ to $\ket{\mathbf{e}V^T \mod 2}$. This enables us to construct a logical $B$ operator $\overline{B}_V$ on the embedded codespace from single qubit phase gates.}\label{fig_embedded_codes}
\end{figure}

\subsection{Canonical Form for Logical Phase Operators}\label{sec_P_canonical}
In the proposition below, we show that logical phase operators have a particularly simple form in terms of the phase-rotation gates of \Cref{sec_phase_rotation}.

\begin{proposition}[Canonical Logical $P$ Operator]\label{prop_logical_P}
Let $\mathbf{z}_i \in \mathbb{Z}_2^n$ be the $Z$-component of a logical $Z_i$ operator $\overline{Z_i}:=XP_2(0|\mathbf{0}|\mathbf{z}_i)$. The operator $\overline{P_i}:=RP_N(2,\mathbf{z}_i)$ is a logical $P_i$ operator.
\begin{proof}
The action of a $P_i$ operator on a computational basis element $\ket{\mathbf{v}}$ where $\mathbf{v}\in \mathbb{Z}_2^k$ can be written $P_i\ket{\mathbf{v}} = \omega^{2\mathbf{v}[i]}\ket{\mathbf{v}}$.
Let $\mathcal{C}:\mathcal{H}_2^k\rightarrow\mathcal{H}_2^n$ be the encoding operator $\mathcal{C}\ket{\mathbf{v}} = \ket{\mathbf{v}}_L$ for $\mathbf{v}\in \mathbb{Z}_2^k$.
From \Cref{eq_LO_condition}, $\overline{P_i}$ is a logical $P_i$ operator if $\overline{P_i}\mathcal{C} = \mathcal{C}P_i$. Hence:
\begin{align}
\mathcal{C}P_i\ket{\mathbf{v}} &= \omega^{2\mathbf{v}[i]}\ket{\mathbf{v}}_L = \sum_{\mathbf{u}\in \mathbb{Z}_2^r}\omega^{2\mathbf{v}[i]}\ket{\mathbf{e_{uv}}} \\&= \overline{P_i}\ket{\mathbf{v}}_L = \sum_{\mathbf{u}\in \mathbb{Z}_2^r}\overline{P_i}\ket{\mathbf{e_{uv}}}.
\end{align}
Hence, we require $\overline{P_i}\ket{\mathbf{e_{uv}}} = \omega^{2\mathbf{v}[i]}\ket{\mathbf{e_{uv}}}$. 
Set the precision $N=2$, and we have $\overline{Z_i}\ket{\mathbf{e_{uv}}} = (-1)^{\mathbf{v}[i]}\ket{\mathbf{e_{uv}}}$. Applying \Cref{eq_XP_action}, we have $\overline{Z_i}\ket{\mathbf{e_{uv}}} = XP_2(0|\mathbf{0}|\mathbf{z}_i)\ket{\mathbf{e_{uv}}} = (-1)^{\mathbf{e_{uv}}\cdot \mathbf{z}_i}\ket{\mathbf{e_{uv}}}$. Therefore, $\mathbf{e_{uv}}\cdot \mathbf{z}_i \mod 2 = \mathbf{v}[i]$.
Now consider the action of $\overline{P_i}:=RP_N(2,\mathbf{z}_i)$ on $\ket{\mathbf{e_{uv}}}$ using \Cref{prop_RP_action}:
\begin{align}
    RP_N(2,\mathbf{z}_i)\ket{\mathbf{e_{uv}}} = \omega^{2(\mathbf{e_{uv}}\cdot\mathbf{z}_i \mod 2)}\ket{\mathbf{e_{uv}}} = \omega^{2\mathbf{v}[i]}\ket{\mathbf{e_{uv}}},
\end{align}
as required for $\overline{P_i}$ to act as a logical $P$ operator.
\end{proof}
\end{proposition}
Using the duality of $\text{RP}$ and $\text{CP}$ operators of \Cref{sec_duality}, we can write $\overline{P}_i$ as a product of CP gates:
\begin{align}\label{eq_canonical_logical_P}
    \overline{P}_i := RP_N(2,\mathbf{z}_i) = \prod_{\mathbf{0} \ne \mathbf{u} \preccurlyeq \mathbf{z}_i}\text{CP}_N( 2\cdot (-2)^{\text{wt}(\mathbf{u})-1},\mathbf{u}).
\end{align}
As $N = 2^t$, any terms with $\text{wt}(\mathbf{u}) > t$ disappear. Hence the support of the CP gates in the implementation are of maximum size $t$. The implementation may not be transversal, as a qubit may be acted upon by more than one gate.

\begin{example}[Logical Phase Operators of the 2D Toric Code]\label{eg_logical_P}
We illustrate the canonical form of logical controlled-phase operators by considering the 2D toric code. Using the XP operator notation of \Cref{eq_XP_def}, let $Z_0 := XP_2(0|\mathbf{0}|\mathbf{z}_0), Z_1 := XP_2(0|\mathbf{0}|\mathbf{z}_1)$ be logical $Z$ operators on logical qubit 0 and 1 respectively with $d := \text{wt}(\mathbf{z}_0) = \text{wt}(\mathbf{z}_1)  \ge 3$. Applying \Cref{eq_canonical_logical_P} and using the notation of \Cref{eq_CP_def} for controlled-phase operators, the canonical forms for the logical $S$ and $T$ operators on qubit 0 are as follows:
\begin{align}
    \overline{S}_0 &:= RP_4(2,\mathbf{z}_0)=\prod_{\mathbf{0} \ne \mathbf{u} \preccurlyeq \mathbf{z}_0}\text{CP}_4(2\cdot  (-2)^{\text{wt}(\mathbf{u})-1},\mathbf{u}) \\
    &=\prod_{\substack{\mathbf{u} \preccurlyeq \mathbf{z}_0\\\text{wt}(\mathbf{u}) =2}}\text{CP}_4(-4,\mathbf{u})\prod_{\substack{\mathbf{u} \preccurlyeq \mathbf{z}_0\\\text{wt}(\mathbf{u}) =1}}\text{CP}_4(2,\mathbf{u})\\
    &= \prod_{i < j \preccurlyeq \mathbf{z}_0}CZ_{ij} \prod_{i \preccurlyeq \mathbf{z}_0}S_i\\
    \overline{T}_0 &:= RP_8(2,\mathbf{z}_0) = \prod_{\mathbf{0} \ne \mathbf{u} \preccurlyeq \mathbf{z}_0}\text{CP}_8(  2\cdot (-2)^{\text{wt}(\mathbf{u})-1},\mathbf{u})\\
    &=\prod_{\substack{\mathbf{u} \preccurlyeq \mathbf{z}_0\\\text{wt}(\mathbf{u}) =3}}\text{CP}_8(8,\mathbf{u})\prod_{\substack{\mathbf{u} \preccurlyeq \mathbf{z}_0\\\text{wt}(\mathbf{u}) =2}}\text{CP}_8(-4,\mathbf{u})\prod_{\substack{\mathbf{u} \preccurlyeq \mathbf{z}_0\\\text{wt}(\mathbf{u}) =1}}\text{CP}_8(2,\mathbf{u})\\
    &= \prod_{i < j < k \preccurlyeq \mathbf{z}_0}CCZ_{ijk} \prod_{i < j \preccurlyeq \mathbf{z}_0}CS^{-1}_{ij} \prod_{i \preccurlyeq \mathbf{z}_0}T_i
\end{align}
These results hold for any CSS code with $\text{wt}(\mathbf{z}_0) \ge 3$, as no other special properties of the toric code have been used.
\end{example}
\subsection{Canonical Form of Logical Phase-Rotation and Controlled-Phase Operators}\label{sec_CP_canonical}
We now generalise the method in \Cref{sec_P_canonical} and show how to implement logical phase-rotation operators for CSS codes using physical phase-rotation gates. Let $L_Z$ be the $k 
\times n$ binary matrix representing logical $Z$ operators such that $L_Z L_X^T \mod 2 = I_k$ where $k =|L_X|$. This means that $XP_2(0|\mathbf{0}|\mathbf{z}_i)$ anti-commutes with $XP_2(0|\mathbf{x}_j|\mathbf{0})$ if and only if $i=j$. Let $\mathbf{u}$ be a binary vector of length $k$. In \Cref{prop_logical_RP} we show that the following is a logical phase-rotation operator:
\begin{align}
\overline{\text{RP}_N(2,\mathbf{u})}:=\text{RP}_N(2,\mathbf{u}L_Z)
\end{align}
By the duality result of \Cref{sec_duality}, we can write logical  phase-rotation operators as follows:
\begin{align}\label{eq_canonical_CP}
    \overline{\text{CP}_N(2^{\text{wt}(\mathbf{v})},\mathbf{v})} &:= \prod_{\mathbf{0} \ne \mathbf{u} \preccurlyeq \mathbf{v} }\overline{\text{RP}_N( 2\cdot (-1)^{\text{wt}(\mathbf{u})-1},\mathbf{u})}\\&=\prod_{\mathbf{0} \ne \mathbf{u} \preccurlyeq \mathbf{v}}\text{RP}_N(2 \cdot (-1)^{\text{wt}(\mathbf{u})-1},\mathbf{u}L_Z)\label{eq_logical_CP_RP}
\end{align}
This in turn can be converted into products of physical controlled-phase gates by applying the duality result a second time.

\begin{framed}
\paragraph{Algorithm 7: Canonical Logical Controlled-Phase Operators}
\paragraph{Input:}
 \begin{enumerate}
     \item The $Z$-logicals $L_Z$ of a CSS code (see above);
     \item A level-$t$ diagonal Clifford hierarchy operator $B$ on $k$ qubits (\Cref{sec_controlled_phase_operators}).
 \end{enumerate}
\paragraph{Output:}
A logical $\overline{B}$ operator (\Cref{sec_LO_of_CSS_codes}) on the code composed of physical phase rotation gates (\Cref{sec_phase_rotation}) with maximum support size $t$.
\paragraph{Method:}
 \begin{enumerate}
     \item Express $B = \prod_{\mathbf{u}} RP_N(\mathbf{q_u, u})$ as a product of phase rotation gates using the duality result of \Cref{sec_duality} where $N = 2^t$ and $\mathbf{u} \in \mathbb{Z}_2^k$;
     \item The operator $\overline{B} = \prod_{\mathbf{u}} RP_N(\mathbf{q_u, u} L_Z)$ is a logical $B$ operator;
     \item Apply the duality result of \Cref{sec_duality} twice to express $\overline{B}$ as a product of phase-rotation gates of maximum support size $t$.
 \end{enumerate}     
\end{framed}

\begin{example}
[Logical Controlled-Phase Operators of Toric Code]
We now demonstrate a canonical implementation of a logical CZ  operator on the 2D toric code of \Cref{eg_logical_P} composed of physical controlled-phase gates. Using \Cref{eq_logical_CP_RP} and the fact that $\text{RP}_4(2,\mathbf{z}_0) = \prod_{i < j \preccurlyeq \mathbf{z}_0}CZ_{ij}\prod_{i \preccurlyeq \mathbf{z}_0}S_i$ from \Cref{eg_logical_P}:
\begin{align}
    \overline{\text{CZ}_{01}}&:= \overline{CP_4(4,11)}\\ 
    &= \prod_{\mathbf{0} \ne \mathbf{u} \preccurlyeq 11}\text{RP}_4(2 \cdot (-1)^{\text{wt}(\mathbf{u})-1},\mathbf{u}L_Z)\\
    &= \text{RP}_4(-2,\mathbf{z}_0 \oplus \mathbf{z}_1)\text{RP}_4(2,\mathbf{z}_0)\text{RP}_4(2,\mathbf{z}_1)\\
    &= \Big(\prod_{i < j \preccurlyeq \mathbf{z}_0\oplus\mathbf{z}_1}CZ_{ij}\prod_{i \preccurlyeq \mathbf{z}_0\oplus\mathbf{z}_1}S_i\Big)^{-1}\Big(\prod_{i < j \preccurlyeq \mathbf{z}_0}CZ_{ij}\prod_{i \preccurlyeq \mathbf{z}_0}S_i\Big)\Big(\prod_{i < j \preccurlyeq \mathbf{z}_1}CZ_{ij}\prod_{i \preccurlyeq \mathbf{z}_1}S_i\Big)\label{eq_CZ_01}
\end{align}
We can choose logical Z operators for the 2D toric code such that $\text{supp}(\mathbf{z}_0) \cap \text{supp}(\mathbf{z}_1) = \emptyset$. In this case, all $S$ operators in \Cref{eq_CZ_01} cancel, as do any CZ operators which lie entirely on the support of either $\mathbf{z}_0$ or $\mathbf{z}_1$, and so we have:
\begin{align}
\overline{\text{CZ}_{01}} &:= \prod_{i \preccurlyeq \mathbf{z}_0, j \preccurlyeq \mathbf{z}_1}CZ_{ij}
\end{align}
This is an instance of Claim 2 in \cite{logical_CZ} for logical multi-controlled-$Z$ operators. Our method applies to arbitrary diagonal Clifford hierarchy logical operators and we can also show:
\begin{align}
\overline{\text{CS}_{01}} &:= \overline{CP_8(4,11)}\\
&=\prod_{\mathbf{0} \ne \mathbf{u} \preccurlyeq 11}\text{RP}_8(2 \cdot (-1)^{\text{wt}(\mathbf{u})-1},\mathbf{u}L_Z)\\
&= \prod_{i \preccurlyeq \mathbf{z}_0, j \preccurlyeq \mathbf{z}_1}CS_{ij}\prod_{i \preccurlyeq \mathbf{z}_0, j < k \preccurlyeq \mathbf{z}_1}CCZ_{ijk}\prod_{i < j \preccurlyeq \mathbf{z}_0,  k \preccurlyeq \mathbf{z}_1}CCZ_{ijk}
\end{align}
Note  that the number of physical gates used in the implementation is $\mathcal{O}(d^t)$. 
 As we are not guaranteed that $\text{supp}(\mathbf{z}_0) \cap \text{supp}(\mathbf{z}_1) = \emptyset$ for arbitrary CSS codes, the above identities are not completely general. In the \href{https://github.com/m-webster/CSSLO/blob/main/notebooks/04.1_canonical_LO.ipynb}{linked Jupyter notebook}, users can calculate identities of this kind for any desired CSS code for any diagonal Clifford hierarchy logical operator.
\end{example}

\subsection{Constructing a CSS Code with a Desired Diagonal Logical Clifford Hierarchy Operator}\label{sec_constructing_css_codes}
In this Section, we apply the canonical logical operator form of \Cref{sec_CP_canonical} to generate a CSS code with a transversal implementation of a desired logical controlled-phase operator using single-qubit phase gates. 
\begin{framed}
\paragraph{Algorithm 8: Constructing CSS Codes with a Desired Diagonal Logical Clifford Hierarchy Operator}
\paragraph{Input:} A controlled-phase operator $B$ on $k$ qubits (\Cref{sec_controlled_phase_operators}) and a target distance $d$. 
\paragraph{Output:} A CSS code with a logical $\overline{B}$ operator (\Cref{sec_LO_of_CSS_codes}) composed of single-qubit phase gates.
\paragraph{Method:}
\begin{enumerate}
    \item Let $\mathbf{C}$ be a $k$-dimensional toric code of distance $d$. We construct the stabiliser generators of $\mathbf{C}$ using the total complex of the tensor product of $k$ classical repetition codes on $d$ bits (see Section II.D of \cite{balanced_product}). The resulting CSS code has $k$ non-overlapping logical Z operators of weight $d$;
    \item Find the canonical implementation of $\overline{B} = \prod_{\mathbf{v}\in V}RP_N(q_\mathbf{v},\mathbf{v})$ composed of phase-rotation gates of maximum support size $t$ using Algorithm 7;
    \item Remove any elements of $V$ where $q_\mathbf{v} = 0$ and apply the embedding $\mathcal{E}_V$ to find the $X$-checks and $X$-logicals of the embedded code $\mathbf{C}_V$ as in \Cref{sec_embedded_codes};
    \item The resulting code has a logical $B$ operator $\overline{B}_V$ composed of level-$t$ phase gates acting on the embedded codespace.
\end{enumerate}
\end{framed}
\begin{example}[Constructing CSS Codes with Transversal Logical Controlled Phase Operators]\label{eg_CS}
In \Cref{table_code_search}, we list the parameters of CSS codes with transversal implementations of various target logical controlled phase operators using the method in \Cref{sec_constructing_css_codes}. 
The CSS codes are generated from toric codes as follows. 
For a  target operator acting on $k$ logical qubits, we use a $k$-dimensional toric code. We generate a series of codes by increasing the distance $d$ of the toric code. 
Looking at the CZ column we have a family of $[[4m^2, 2, 2m]]$ codes with a transversal CZ operator, the first member of which is the $[[4,2,2]]$ code of \Cref{eg_CSS_LO}. 
Looking at the CCZ column, we have a family of $[[8m^3, 3, 2m]]$ codes which have a transversal CCZ operator, the first member of which is the hypercube code of \Cref{eg_hypercube_LI}. 
The 6-qubit code in the S column is the 6-qubit code discussed in \Cref{eg_embedding_rep}.
The 15-qubit code in the T column is the 15-qubit Reed-Muller code.
The first entry in the CS column is the $[[12,2,2]]$ code with the following $X$-checks and $X$-logicals:
\begin{align}
S_X&:=\begin{pmatrix}
111100001111\\
000011111111
\end{pmatrix}; &
L_X&:=\begin{pmatrix}
010101010101\\
001100110011
\end{pmatrix}.
\end{align}

An interactive version is available in the \href{https://github.com/m-webster/CSSLO/blob/main/notebooks/04.2_code_search.ipynb}{linked Jupyter notebook}.
\begin{table}[hbt]
	\centering

\begin{tabular}{|r||r|r|r|r|r|r|r|r|r|r|r|r|r|r|r|}
\hline
\textbf{d} &\multicolumn{3}{|l|}{\textbf{Logical S}} &\multicolumn{3}{l|}{\textbf{Logical CZ}} &\multicolumn{3}{l|}{\textbf{Logical T}} &\multicolumn{3}{l|}{\textbf{Logical CS}}&\multicolumn{3}{l|}{\textbf{Logical CCZ}}\\
\hline
	&\textbf{n}	&$\mathbf{d_X}$	&$\mathbf{d_Z}$	&\textbf{n}	&$\mathbf{d_X}$	&$\mathbf{d_Z}$		&\textbf{n}	&$\mathbf{d_X}$	&$\mathbf{d_Z}$	&\textbf{n}	&$\mathbf{d_X}$	&$\mathbf{d_Z}$		&\textbf{n}	&$\mathbf{d_X}$	&$\mathbf{d_Z}$	\\
\hline
\textbf{2}	&1	&1	&1	&4	&2	&2	&1	&1	&1	&12	&6	&2	&8	&4	&2	\\
\textbf{3}	&6	&3	&2	&15	&4	&3	&1	&1	&1	&33	&14	&2	&63	&16	&3	\\
\textbf{4}	&6	&3	&2	&16	&4	&4	&14	&7	&2	&64	&22	&2	&64&16	&4	\\
\textbf{5}	&15	&5	&3	&35	&6	&5	&15	&7	&3	&155&40	&3	&215&36	&5	\\
\textbf{6}	&15	&5	&3	&36	&6	&6	&35	&15	&2	&228&52	&4	&216&36	&6	\\
\textbf{7}	&28	&7	&4	&63	&8	&7	&36	&15	&3	&385&76	&4	&511&64	&7	\\
\textbf{8}	&28	&7	&4	&64	&8	&8	&92	&29	&3	&512&92	&5	&512&64	&8	\\
\textbf{9}	&45	&9	&5	&99	&10	&9	&93 &29	&3	&819&126&5	&999&100&9	\\
\textbf{10}	&45	&9	&5	&100&10	&10	&165&45	&4	&1020&146&6	&1000&100&10	\\
\hline
\end{tabular}

\caption{Parameters of CSS codes generated by the embedded code method when searching for implementations of logical operators based on the toric code of distance $d$. For a logical operator acting on $k$ qubits, we use a $k$-dimensional toric code.} \label{table_code_search}
\end{table}

\begin{figure}[hbt!]
\centering
\includegraphics[width=0.9\textwidth]{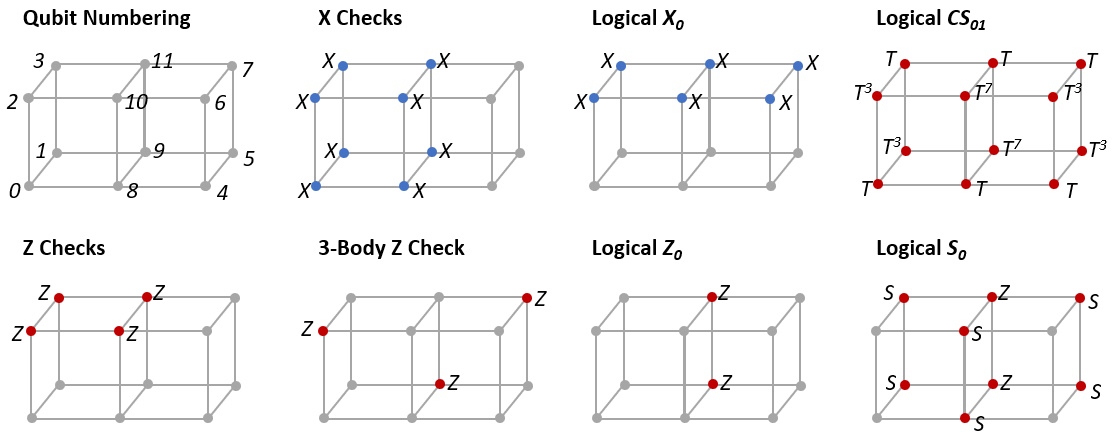}

\caption{\textbf{CSS Code with Transversal Logical Controlled-$S$ Operator}: the $[[12,2,2]]$ code of \Cref{eg_CS} is formed from two $[[8,3,2]]$ hypercube codes of \Cref{eg_hypercube_LI} joined at a common face, with additional 3-body Z-checks. Similarly, the $[[14,1,2]]$ code with a transversal logical $T$ operator is formed from three $[[8,3,2]]$ codes joined pairwise at  faces and sharing a common edge.}\label{fig_CS}
\end{figure}

\end{example}

\section{Conclusion and Open Questions}
We have presented efficient new methods to identify and test diagonal logical operators on CSS codes using both single- and multi-qubit diagonal Clifford hierarchy gates as building blocks.  In addition, we provided a technique for generating CSS codes with implementations of any desired diagonal Clifford hierarchy logical operator using single-qubit phase gates. The methods generalise to non-CSS stabiliser codes as demonstrated in \Cref{app_non_css}. The algorithms are available in a \href{https://github.com/m-webster/CSSLO}{GitHub repository} and are intended to be of benefit to researchers in understanding the logical operator structure of stabiliser codes.

Our methods rely on representing diagonal Clifford hierarchy operators as diagonal XP operators. Our algorithms use the vector representation of XP operators and linear algebra modulo $N$, and so have reduced computational complexity compared to existing work in this area. 

The ability to represent diagonal Clifford hierarchy operators as XP operators may have a number of other possible applications. 
Custom design of CSS codes for devices that have known coherent noise models is one possibility. If the noise can be represented as a series of multi-qubit diagonal operators, we could design a CSS code where these operators are in the logical identity group and so mitigate coherent noise. 
The simulation of quantum circuits could be another application. A circuit composed of multi-qubit diagonal operators, such as those used for measuring the stabiliser generators of a CSS code, could be amenable to simulation using XP representations of the gates used. 
As any diagonal Clifford hierarchy operator can be represented as a diagonal XP operator, there could also be implications for computational complexity theory.

\begin{acknowledgements} 
This research was supported by the Australian Research Council via the Centre of Excellence in Engineered Quantum Systems (EQUS) project number CE170100009. MW was supported by the Sydney Quantum Academy, Sydney, NSW, Australia. MW and AOQ were students at the 2022 IBM Quantum Error Correction Summer School and the idea for this paper originated there; we wish to express our thanks to IBM for the opportunity to attend the Summer School. MW thanks Robert Calderbank for hosting his visit to Duke University in August 2022 -- the basic structure of this paper was developed during the visit.
\end{acknowledgements}

\appendix

\section{Controlled-Phase and Phase-Rotation Operators}\label{app_CP_RP}

In this Appendix, we give the detailed proofs of results relating to controlled-phase and phase-rotation operators. The action  and duality property of these operators derive from sum/product duality properties for binary vectors and binary variables, and we start by proving these results. We then prove results relating to phase-rotation operators. We first show that phase-rotation operators can be written as a sum of projectors. This allows us to calculate the logical action of phase-rotation operators. We then prove the duality result between controlled-phase and  phase-rotation operators.  Finally we prove the key commutation relations for phase-rotation and controlled-phase operators. 

\subsection{Product/Sum Duality Results for Binary Vectors and Variables}

\begin{proposition}[Sum/Product Duality of Binary Vectors]\label{prop_product_sum}
Let $L$ be a binary matrix with rows $\mathbf{x}_i$ for $0 \le i < r$ and $\mathbf{v}$ a binary vector of length $r$. Define:
\begin{align}
    \mathbf{s_v}(L)&:= \bigoplus_{i \preccurlyeq \mathbf{v}}\mathbf{x}_i;\\
    \mathbf{p_v}(L)&:= \prod_{i \preccurlyeq \mathbf{v}}\mathbf{x}_i.
\end{align}
Then over the integers:
\begin{align}
    \mathbf{s_v}(L)&= \sum_{\mathbf{0} \ne \mathbf{u} \preccurlyeq \mathbf{v}}  (-2)^{\text{wt}(\mathbf{u})-1} \mathbf{p_u}(L);\\
    2^{\text{wt}(\mathbf{v})-1}\mathbf{p_v}(L)&= \sum_{\mathbf{0} \ne \mathbf{u} \preccurlyeq \mathbf{v}} 
     (-1)^{\text{wt}(\mathbf{u})-1} \mathbf{s_u}(L).\label{eq_product_sum}
\end{align}
\end{proposition}
\begin{proof}
Restatement of Proposition E.10 of \cite{XP}.
\end{proof}

\begin{proposition}[Sum/Product Duality of Binary Variables]\label{prop_product_sum_variables}
Let $\mathbf{e}$ be a vector of $r$ binary variables and let $\mathbf{v}$ a binary vector of length $r$. Define:
\begin{align}
    s_\mathbf{v}(\mathbf{e})&:= \bigoplus_{i \preccurlyeq \mathbf{v}}\mathbf{e}[i];\\
    p_\mathbf{v}(\mathbf{e})&:= \prod_{i \preccurlyeq \mathbf{v}}\mathbf{e}[i].
\end{align}
Then over the integers:
\begin{align}
    s_\mathbf{v}(\mathbf{e})&= \sum_{\mathbf{0} \ne \mathbf{u} \preccurlyeq \mathbf{v}} (-2)^{\text{wt}(\mathbf{u})-1} p_\mathbf{u}(\mathbf{e});\\
    2^{\text{wt}(\mathbf{v})-1} p_\mathbf{v}(\mathbf{e}) &= \sum_{\mathbf{0} \ne \mathbf{u} \preccurlyeq \mathbf{v}} (-1)^{\text{wt}(\mathbf{u})-1} s_\mathbf{u}(\mathbf{e}).
\end{align}
\end{proposition}
\begin{proof}
Application of \Cref{prop_product_sum} with $L$ the single-column matrix $\mathbf{e}^T$.
\end{proof}

\subsection{Phase-Rotation Operators}

\begin{proposition}[Projector Form of RP Operators]\label{prop_RP_projector}
Phase-rotation operators can be written in terms of projectors $A_{\pm 1} := (I \pm A)/2$:
\begin{align}
    \text{RP}_N(q,\mathbf{v}) := \exp(\frac{q\pi i}{N} A_{-1}) = A_{+1} + \omega^q A_{-1}.
\end{align}
\end{proposition}
\begin{proof}
    Because $A_{-1}$ is a projector, $A_{-1}^m =  A_{-1}$ for integers $m > 0$. Also  $A_{-1}^0 = I = A_{+1} + A_{-1}$. Hence:
\begin{align}
\exp((q\pi i/N)A_{-1})&=I + A_{-1}\sum_{m >0}(q\pi i/N)^m/m!\\
&=A_{+1} + A_{-1}\sum_{m \ge 0}(q\pi i/N)^m/m!\\&=A_{+1} + e^{q\pi i/N}A_{-1}
\end{align}
\end{proof}

\begin{proposition}[Action of RP Operators]\label{prop_RP_action}
The action of a phase-rotation operator on a computational basis element $|\mathbf{e}\rangle$ where $\mathbf{e} \in \mathbb{Z}_2^n$ and $\omega := \exp(\pi i/N)$ is:
\begin{align}
    \text{RP}_N(q,\mathbf{v})|\mathbf{e}\rangle = \begin{cases}\omega^{q}|\mathbf{e}\rangle & \text{ if }\mathbf{e}\cdot\mathbf{v} \mod 2  = 1\\|\mathbf{e}\rangle & \text{ otherwise}.\end{cases}
\end{align}
\end{proposition}
\begin{proof}
Straightforward application of projector form of phase-rotation operators in \Cref{prop_RP_projector}.
\end{proof}

\subsection{Duality of Controlled-Phase and Phase-Rotation Operators}
The proposition below allows us to express controlled-phase operators (\Cref{sec_controlled_phase_operators}) as products of phase-rotation operators (\Cref{sec_phase_rotation}) and vice-versa.
\begin{proposition}[Duality of Controlled-Phase and Phase-Rotation Operators]\label{prop_duality}
For $N= 2^t$ and $\mathbf{u, v}$ binary vectors of length $n$ the following identities hold:
\begin{align}
    \text{RP}_N(2,\mathbf{v}) = \prod_{\mathbf{0} \ne \mathbf{u} \preccurlyeq \mathbf{v}}\text{CP}_N(2 \cdot(-2)^{\text{wt}(\mathbf{u})-1}  , \mathbf{u})\\
    \text{CP}_N(2^{\text{wt}(\mathbf{v})},\mathbf{v}) = \prod_{\mathbf{0} \ne \mathbf{u} \preccurlyeq \mathbf{v}}\text{RP}_N(2\cdot (-1)^{\text{wt}(\mathbf{u})-1}, \mathbf{u})
\end{align}
\end{proposition}
\begin{proof}
Using \Cref{eq_RP_phase} and the notation of \Cref{prop_product_sum_variables}, we can write $\text{RP}_N(2,\mathbf{v})\ket{\mathbf{e}} = \omega^{2s_\mathbf{v}(\mathbf{e})}\ket{\mathbf{e}}$. From  \Cref{prop_product_sum_variables}, we have $2s_\mathbf{v}(\mathbf{e}) = \sum_{\mathbf{0} \ne \mathbf{u} \preccurlyeq \mathbf{v}}2\cdot(-2)^{\text{wt}(\mathbf{u})-1}p_\mathbf{u}(\mathbf{e})$.

Similarly, from \Cref{eq_CP_phase}, we can write $\text{CP}_N(2^{\text{wt}(\mathbf{v})},\mathbf{v})\ket{\mathbf{e}} = \omega^{2^{\text{wt}(\mathbf{v})}p_\mathbf{v}(\mathbf{e})}\ket{\mathbf{e}}$ and due to \Cref{prop_product_sum_variables}, we have $2^{\text{wt}(\mathbf{v})} p_\mathbf{v}(\mathbf{e}) = \sum_{\mathbf{0} \ne \mathbf{u} \preccurlyeq \mathbf{v}}2\cdot (-1)^{\text{wt}(\mathbf{u})-1} s_\mathbf{u}(\mathbf{e})$.

Hence the phases applied on the RHS and LHS are the same and the result follows.
\end{proof}

\subsection{Commutator Relations for Controlled-Phase and Phase-Rotation Operators}
In this Section, we prove the commutation relations for Pauli X operators with controlled-phase operators (\Cref{sec_controlled_phase_operators}) and phase-rotation operators (\Cref{sec_phase_rotation}).
\begin{proposition}[Commutator Relation for Phase-Rotation Operators]\label{prop_RP_algebra}
Let $X_i$ denote a Pauli $X$ operator on qubit $i$. The following identity applies for phase-rotation operators:
    \begin{align}
    \text{RP}_N(q,\mathbf{v})X_i = \begin{cases} \omega^q X_i\text{RP}_N(-q,\mathbf{v})  & \text{ if } \mathbf{v}[i] = 1\\  X_i \text{RP}_N(q,\mathbf{v})& \text{ otherwise}\end{cases}
\end{align}
\end{proposition}
\begin{proof}
If $\mathbf{v}[i] =0$,  the support of the operators do not overlap and hence the operators commute and the second case follows.

For the case where $\mathbf{v}[i] = 1$, let $\mathbf{b}^n_i$ be the length $n$ vector which is zero apart from component $i$ which is one. Then, for a computational basis vector $|\mathbf{e}\rangle$, we have:
\begin{align}
    \text{RP}_N(q,\mathbf{v})X_i|\mathbf{e}\rangle &= \text{RP}_N(q,\mathbf{v})|\mathbf{e} \oplus \mathbf{b}^n_i\rangle\\
    &= \begin{cases}|\mathbf{e} \oplus \mathbf{b}^n_i\rangle & \text{ if }(\mathbf{e} \oplus \mathbf{b}^n_i) \cdot \mathbf{v} = 0 \mod 2\\\omega^q |\mathbf{e} \oplus \mathbf{b}^n_i\rangle & \text{ otherwise}.\end{cases}\\
    \omega^q X_i\text{RP}_N(-q,\mathbf{v})|\mathbf{e}\rangle &= \begin{cases}\omega^q |\mathbf{e} \oplus \mathbf{b}^n_i\rangle & \text{ if }\mathbf{e} \cdot \mathbf{v} = 0 \mod 2\\\omega^q\omega^{-q} |\mathbf{e} \oplus \mathbf{b}^n_i\rangle & \text{ otherwise}.\end{cases}
\end{align}
Since, by assumption $\mathbf{v}[i] =1$, $\mathbf{e} \cdot \mathbf{v} = 0 \mod 2 \iff (\mathbf{e} \oplus \mathbf{b}^n_i) \cdot \mathbf{v} = 1 \mod 2$. Hence, the action on computational basis vectors is identical and the result follows.
\end{proof}

\begin{proposition}[Commutation Relation for Controlled-Phase Operators]\label{prop_CP_relation}
\begin{align}
\text{CP}_N(q,\mathbf{v}) X_i  = \begin{cases}X_i\text{CP}_N(q, \mathbf{v}) & \text{ if } \mathbf{v}[i] = 0\\X_i\text{CP}_N(-q,\mathbf{v})\text{CP}_N(q,\mathbf{v}\oplus \mathbf{b}^n_i) & \text{ otherwise}\end{cases}
\end{align}
Where $\mathbf{b}^n_i$ is the length $n$ binary vector which is zero apart from component $i$ which is one.
\end{proposition}
\begin{proof}
If $\mathbf{v}[i] = 0$ then $\text{CP}_N(q,\mathbf{v})$ has no support in common with $ X_i$ so the operators commute. Now assume $\mathbf{v}[i] = 1$ then the operator on the LHS acts on the  computational basis element $|\mathbf{e}\rangle$ as follows:
\begin{align}
    \text{CP}_N(q,\mathbf{v}) X_i|\mathbf{e}\rangle &= \text{CP}_N(q,\mathbf{v}) |\mathbf{e} \oplus \mathbf{b}^n_i\rangle\\
    &= \begin{cases}\omega^q|\mathbf{e} \oplus \mathbf{b}^n_i\rangle & \text{ if } \mathbf{v} \preccurlyeq (\mathbf{e} \oplus \mathbf{b}^n_i)\\|\mathbf{e} \oplus \mathbf{b}^n_i\rangle & \text{ otherwise}.\end{cases}
\end{align}
Since by assumption $\mathbf{v}[i] = 1$, a phase of $\omega^q$ is applied $\iff \mathbf{v} \preccurlyeq (\mathbf{e} \oplus \mathbf{b}^n_i) \iff \mathbf{e}[i] = 0\text{ AND } (\mathbf{v} \oplus \mathbf{b}^n_i) \preccurlyeq \mathbf{e}$. Now consider the RHS and assume $\mathbf{e}[i] = 0 \text{ AND } (\mathbf{v} \oplus \mathbf{b}^n_i) \preccurlyeq \mathbf{e}$. In this case, we do not have $\mathbf{v}\preccurlyeq \mathbf{e}$ because $\mathbf{v}[i] = 1$ but $\mathbf{e}[i] = 0$. Hence:
\begin{align}
    X_i\text{CP}_N(-q,\mathbf{v})\text{CP}_N(q,\mathbf{v}\oplus \mathbf{b}^n_i)|\mathbf{e}\rangle &= \omega^{q}X_i\text{CP}_N(-q,\mathbf{v})|\mathbf{e}\rangle\\
    &= \omega^{q}X_i|\mathbf{e}\rangle = \omega^{q}|\mathbf{e}\oplus \mathbf{b}^n_i\rangle 
\end{align}
We now show that all other cases result in a trivial phase. Assume $\mathbf{e}[i] = 1 \text{ AND } (\mathbf{v} \oplus \mathbf{b}^n_i) \preccurlyeq \mathbf{e}$. In this case, $\mathbf{v}\preccurlyeq \mathbf{e}$ and so:
\begin{align}
    X_i\text{CP}_N(-q,\mathbf{v})\text{CP}_N(q,\mathbf{v}\oplus \mathbf{b}^n_i)|\mathbf{e}\rangle &= \omega^{q}X_i\text{CP}_N(-q,\mathbf{v})|\mathbf{e}\rangle\\
    &= \omega^{q}\omega^{-q}X_i|\mathbf{e}\rangle = |\mathbf{e}\oplus \mathbf{b}^n_i\rangle 
\end{align}

Now assume that $(\mathbf{v} \oplus \mathbf{b}^n_i) \preccurlyeq \mathbf{e}$ is not true. In this case, we can never have $\mathbf{v} \preccurlyeq \mathbf{e}$ and so neither of the controlled-phase operators apply a phase, regardless of the value of $\mathbf{e}[i]$. Hence the LHS and RHS have the same action on computational basis elements and the result follows.
\end{proof}

\begin{example}[Commutation Relation for Controlled-Phase Operators]\label{eg_CP_conj}
Using \Cref{prop_CP_relation}, we can conjugate controlled-phase operators by strings of X operators and vice versa. We first compute $\text{CS}_{01}X_1\text{CS}_{01}^{-1}$ where $\text{CS}_{01}$ is a controlled-$S$ operator on qubits 0 and 1. Using the notation of \Cref{eq_CP_def}:
\begin{align}
\text{CS}_{01}X_1\text{CS}_{01}^{-1} &= \text{CP}_8(4, 11) X_1 \text{CP}_8(-4, 11)\\
&= X_1 \text{CP}_8(-4, 11) \text{CP}_8(4, 10) \text{CP}_8(-4, 11)\\
&= X_1 \text{CP}_8(-8, 11) \text{CP}_8(4, 10)\\
&= X_1 \text{CZ}_{01}S_0
\end{align}
We now compute $(X_0X_1X_2) \text{CCZ}_{012} (X_0X_1X_2)^{-1}$. Using \Cref{eq_CP_relation} with $\mathbf{x} = \mathbf{v} = 111$, and letting $\mathbf{w} := \mathbf{u} \oplus \mathbf{v}$:
\begin{align}
    (X_0X_1X_2) \text{CCZ}_{012} (X_0X_1X_2)^{-1} &= 
    XP_2(0|111|\mathbf{0})\text{CP}_8(8, 111) XP_2(0|111|\mathbf{0})\\
    &=  \prod_{\mathbf{u} \preccurlyeq 111}\text{CP}_8(8 \cdot (-1)^{3+\text{wt}(\mathbf{u})},\mathbf{v}\oplus\mathbf{u})\\
    &= \prod_{0 \le \text{wt}(\mathbf{w}) \le 3}\text{CP}_8(8,\mathbf{w})\\
     &= CP_8(8, \mathbf{0}) \prod_{0 < \text{wt}(\mathbf{w}) \le 3}\text{CP}_8(8,\mathbf{w})\\
     &= -Z_0Z_1Z_2\text{CZ}_{01}\text{CZ}_{02}\text{CZ}_{12}\text{CCZ}_{012}
\end{align}
Interactive versions of these examples are available in the \href{https://github.com/m-webster/CSSLO/blob/main/notebooks/10.2_CP_commutation.ipynb}{linked Jupyter notebook}.
\end{example}

\section{Additional Details for Logical Operator Algorithms}

This Appendix provides further details on the various logical operator algorithms.  
We first prove results that reduce the complexity of the logical action and logical operator test algorithms of \Cref{sec_ker_search,sec_logical_action,sec_LO_test}. 
We then show how to calculate valid $\mathbf{z}$ vectors that result in diagonal operators that commute with the $X$-checks up to a logical identity for use in \Cref{sec_comm_LO}. 
We then demonstrate a method for more efficiently searching for depth-one logical operators composed of multi-qubit controlled phase gates for use in \Cref{sec_depth_one_algorithm}.
We then show that the embedding operator of \Cref{sec_embedded_codes} acts as a group homomorphism on the group generated by phase-rotation and Pauli $X$ operators. 
Finally, we show that the canonical form of \Cref{sec_CP_canonical} results in a logical operator with the required action.

\subsection{Reducing Complexity of Logical Action Algorithms}

In this Section we show how to reduce the complexity of algorithms which work with the logical action of diagonal XP operators on the canonical codewords.
If $B:= XP_N(0|\mathbf{0|z})$ is a diagonal logical operator of precision $N:=2^t$, then the action of $B$ on the computational basis vectors $\mathbf{e_{uv}} := \mathbf{u}S_X + \mathbf{v}L_X$ making up the canonical codewords of \Cref{eq_codewords} can be written as $B\ket{\mathbf{e_{uv}}} = \omega^{2\mathbf{e_{uv}} \cdot \mathbf{z}}\ket{\mathbf{e_{uv}}}$. In the proposition below, we show that the phase component $\mathbf{e_{uv}} \cdot \mathbf{z}$ is completely determined by terms of form $\mathbf{e_{u'v'}} \cdot \mathbf{z}$ where $\text{wt}(\mathbf{u'}) + \text{wt}(\mathbf{v'}) \le t$. As a result, when working with logical actions, we do not need to consider all $2^{k+r}$ of the $\mathbf{e_{uv}}$ vectors, just a limited set which is of size polynomial in $k$ and $r$. This reduces the computational complexity of the algorithms in  \Cref{sec_ker_search} and \Cref{sec_logical_action}. 

\begin{proposition}\label{prop_action_weight}
Let $N := 2^t$ and $\mathbf{z} \in \mathbb{Z}_N^n$. The phase component $\mathbf{e_{uv}} \cdot \mathbf{z}$ can be expressed as a $\mathbb{Z}_N$ linear combination of $\mathbf{e_{u'v'}} \cdot \mathbf{z}$ where $\text{wt}(\mathbf{u'}) + \text{wt}(\mathbf{v'}) \le t$.

\end{proposition}
\begin{proof}
Let $G_X:= \begin{pmatrix}S_X\\L_X\end{pmatrix}$ and let $\mathbf{a}:=(\mathbf{u|v})$. 
Noting that $\mathbf{e_{uv}} = \mathbf{a}G_X = \mathbf{s_a}(G_X)$ and applying \Cref{prop_product_sum} we have:
\begin{align}
    \mathbf{a}G_X = \mathbf{s_a}(G_X)  &= \sum_{\mathbf{0} \ne \mathbf{b} \preccurlyeq \mathbf{a}} (-2)^{\text{wt}(\mathbf{b})-1} \mathbf{p_b}(G_X)
\end{align}
Terms with $\text{wt}(\mathbf{b}) > t$ disappear modulo $N = 2^t$. Using the linearity of dot product and expressing the $\mathbf{p_b}(G_X)$ in terms of $\mathbf{s_c}(G_X)$:
\begin{align}
\mathbf{a}G_X \cdot \mathbf{z} \mod N &= \Big(\sum_{\substack{\mathbf{0} \ne \mathbf{b}\preccurlyeq \mathbf{a}\\\text{wt}(\mathbf{a})\le t}}(-2)^{\text{wt}(\mathbf{b})-1}\mathbf{p_b}(G_X)\Big) \cdot \mathbf{z} \mod N\\
&= \Big(\sum_{\substack{\mathbf{0} \ne \mathbf{b}\preccurlyeq \mathbf{a}\\\text{wt}(\mathbf{a})\le t}}(-2)^{\text{wt}(\mathbf{b})-1}\sum_{\mathbf{0} \ne \mathbf{c}\preccurlyeq \mathbf{b}}(-1)^{\text{wt}(\mathbf{c})-1}\mathbf{s_c}(G_X)\Big) \cdot \mathbf{z} \mod N\\
&= \Big(\sum_{\substack{\mathbf{0} \ne \mathbf{b}\preccurlyeq \mathbf{a}\\\text{wt}(\mathbf{a})\le t}}(-2)^{\text{wt}(\mathbf{b})-1}\sum_{\mathbf{0} \ne \mathbf{c}\preccurlyeq \mathbf{b}}(-1)^{\text{wt}(\mathbf{c})-1}\mathbf{c}G_X\cdot \mathbf{z} \Big) \mod N
\end{align}
As $\text{wt}(\mathbf{c}) \le \text{wt}(\mathbf{b}) \le \text{wt}(\mathbf{a}) \le t$, the result follows.
\end{proof}

\subsection{Test for Diagonal Logical XP Operators}
In this Section, we prove that the algorithm in \Cref{sec_LO_test} correctly identifies diagonal logical operators of XP form. We first show that if the group commutator of an operator $B$ with each of the logical identities $A_1, A_2$ is a logical identity, then the group commutator of the product $A_1A_2$ is a logical identity. 
 

\begin{proposition}[Commutators of Logical Identities]\label{prop_comm_X_checks}
Let $\mathcal{I}_\text{XP}$ be the logical identity group as defined in \Cref{sec_XP_groups} and let $A_1, A_2 \in \mathcal{I}_\text{XP}$. Let $B$ be an XP operator such that $[[A_1, B]]$ and $[[A_2, B]] \in \mathcal{I}_\text{XP}$. Then $[[A_1A_2, B]] \in\mathcal{I}_\text{XP}$
\end{proposition}
\begin{proof}
As $\mathcal{I}_\text{XP}$ is a group, $A_1, A_2 \in \mathcal{I}_\text{XP} \implies [[A_1, A_2]] \in \mathcal{I}_\text{XP}$. Hence we can write $A_1 A_2 = C A_2A_1$ for some $C \in \mathcal{I}_\text{XP}$.
Calculating $[[A_1A_2, B]]$, for some $C, C', C'' \in\mathcal{I}_\text{XP}$:
\begin{align}
[[A_1A_2, B]] &= A_1A_2BA_2^{-1}A_1^{-1}B^{-1}\\
&= A_1(A_2BA_2^{-1}B^{-1}) B A_1^{-1}B^{-1}\\
&= A_1 C B A_1^{-1}B^{-1}\\ 
&= C' A_1  B A_1^{-1}B^{-1}\\
&= C' C'' \in \mathcal{I}_\text{XP}
\end{align}
\end{proof}
We now show that for a diagonal XP operator $B$, it is sufficient to check group commutators with the $r := |S_X|$  operators of form $XP_N(0|\mathbf{x}_i|\mathbf{0})$ where $\mathbf{x}_i$ are the rows of the $X$-checks $S_X$.
\begin{proposition}
Let $\mathbf{C}$ be a CSS code with $X$-checks $S_X$ and logical identity XP group $\mathcal{I}_{XP}$ of precision $N$.
    Let $B:= XP_N(0|\mathbf{0|z})$ be a diagonal XP operator.  $B$ is a logical operator if and only if $[[XP_N(0|\mathbf{x}_i|\mathbf{0}),B]] \in \mathcal{I}_\text{XP}$ for all rows $\mathbf{x}_i$ of $S_X$.
    \begin{proof}
$B$ is a logical operator if and only if $[[A,B]] \in \mathcal{I}_\text{XP}$ for all $A \in \mathcal{I}_\text{XP}$. If $B$ is a logical operator then $[[XP_N(0|\mathbf{x}_i|\mathbf{0}),B]] \in \mathcal{I}_\text{XP}$ because $XP_N(0|\mathbf{x}_i|\mathbf{0}) \in \mathcal{I}_\text{XP}$. 

Conversely, assume $[[XP_N(0|\mathbf{x}_i|\mathbf{0}),B]] \in \mathcal{I}_\text{XP}$ for all rows $\mathbf{x}_i$ of $S_X$. Let $K_M$ be the matrix whose rows are a generating set of the $Z$-components of the logical identities as defined in \Cref{sec_LI}. Any logical identity $A$ can be written as a product of terms of form $XP_N(0|\mathbf{x_i|0})$ and $XP_N(0|\mathbf{0}|\mathbf{z}_j)$ where $\mathbf{z}_j$ is a row of $K_M$.
By assumption, $[[XP_N(0|\mathbf{x}_i|\mathbf{0}),B]] \in \mathcal{I}_\text{XP}$ and $[[XP_N(0|\mathbf{0}|\mathbf{z}_j),B]] = I$. Due to \Cref{prop_comm_X_checks}, the commutator of the product is a logical identity and the result follows.
\end{proof}
\end{proposition}

\subsection{Algorithm to Determine Commutators of a Given $X$-Check}\label{sec_comm_x}
In the method of \Cref{sec_comm_LO}, for a given $X$-check $\mathbf{x} \in S_X$ we seek all $Z$-components $\mathbf{z} \in \mathbb{Z}_N^n$ such that the group commutator $[[XP_N(0|\mathbf{0|z}),XP_N(0|\mathbf{x|0})]]$ is a logical identity. This reduces to solving for $\mathbf{z}$ such that both $\mathbf{x\cdot z} = 0 \mod N$ and $2\mathbf{xz} \in \langle K_M\rangle_{\mathbb{Z}_N}$ where the rows of $K_M$ are the $Z$-components of the diagonal logical identities as in \Cref{sec_LI}. In this Section, we show how to solve for these constraints using linear algebra modulo $N$. The method is as follows:

Without loss of generality, reorder qubits so that the first $m$ components of $\mathbf{x}$ are one and the remaining $n-m$ components are zero. In the matrices of form $(\mathbf{a|b})$ below, the first component has $m$ columns corresponding to the non-zero components of $\mathbf{x}$, the next $n-m$ columns correspond to the zero components of $\mathbf{x}$. For $\mathbf{v} \in \mathbb{Z}_N$, let $\mathbf{v \cdot 1} := (\sum_{0 \le i < n} \mathbf{v}[i]) \mod N$.
\begin{enumerate}
    \item The vector $2\mathbf{xz}$ is of the form $(2\mathbf{u|0})$ where $\mathbf{u}$ is of length $m$ and the row span of  $C_0 := (2I_m | 0 )$ over $\mathbb{Z}_N$ represents all vectors of this form;
    \item Group commutators which are also logical identities are in $\langle C_0 \rangle_{\mathbb{Z}_N} \bigcap \langle K_M \rangle_{\mathbb{Z}_N}$ and a Howell basis $C_1$ is calculated via the intersection of spans method in Appendix 4.1 of \cite{XP};
    \item The rows of $C_1$ are of form $(\mathbf{u|0}) \in \langle K_M \rangle_{\mathbb{Z}_N}$ for $\mathbf{u}$ divisible by $2$ modulo $N$. Now let $\mathbf{v} := \mathbf{u}/2$. Because $2(\mathbf{v} + N/2) = 2\mathbf{v} = \mathbf{u} \mod N$,  $(\mathbf{v}\cdot \mathbf{1}) \mod N$ is either $0$ or $N/2$. Adjust the $m$th component of $\mathbf{v}$ by subtracting $(\mathbf{v}\cdot \mathbf{1}) \mod N$. Let $C_2$ be the matrix formed from rows of form $(\mathbf{v|0})$;
    \item Adding pairs of $N/2$ to the first $m$ components does not change $2\mathbf{x}\mathbf{z}$ or $\mathbf{x}\cdot \mathbf{z} \mod N$. Let $A$ be $I_{m-1}$ with a column of ones appended. Add the rows $(N/2 \cdot A|\mathbf{0})$ to $C_2$;
    \item Columns $i$ where $\mathbf{x}[i] = 0$ can have arbitrary values, as these do not contribute to $2\mathbf{x}\mathbf{z}$ or $\mathbf{x}\cdot \mathbf{z}$. Add the rows $(\mathbf{0}|I_{n-m})$ to $C_2$;
    \item Return qubits to their original order. The valid $Z$-components are given by the row span of $C_2$ over $\mathbb{Z}_N$.
\end{enumerate}

\subsection{Algorithm for Depth-One Operators}\label{app_depth_one_algorithm}
In the depth-one algorithm, we find transversal logical operators by starting with a level-$t$ logical XP operator acting on the embedded codespace, then multiplying by all possible elements of the diagonal logical XP group of the embedded code. If the order of the diagonal logical XP group is large, this method can be computationally expensive. In this Section, we demonstrate an algorithm for more efficiently exploring the search space and checking if an operator acting on the embedded codespace acts transversally on the codespace.  We use the residue function defined in Eq. 142 of \cite{XP} - we say that $\mathbf{z}' = \text{Res}_{\mathbb{Z}_N}(K_L,\mathbf{z})$ if:
\begin{align}
    \text{How}_{\mathbb{Z}_N}\begin{pmatrix}
        1 & \mathbf{z}\\
        0 & K_L
    \end{pmatrix} = \begin{pmatrix}
        1 & \mathbf{z}'\\
        0 & K_L'
    \end{pmatrix}
\end{align}
The input to this algorithm is the following:
\begin{itemize}
    \item A binary matrix $V$ for the embedding operator $\mathcal{E}_V$;
    \item A matrix $K_L$ representing the $Z$-components of the generators of the diagonal logical XP group of the embedded code (see \Cref{sec_comm_LO}).
    \item A row vector $\mathbf{z}$ of $K_L$ which represents the $Z$-component of a non-trivial logical operator at level $t$ of the diagonal Clifford hierarchy acting on the embedded codespace. This corresponds to a product of phase-rotation gates acting on the original codespace;

\end{itemize}
The output is a depth-one implementation of a non-trivial logical operator at level $t$ of the diagonal Clifford hierarchy, or \texttt{FALSE} if no such operator exists.
The algorithm method is as follows:
\begin{enumerate}
    \item Remove $\mathbf{z}$ from $K_L$;
    \item Let \texttt{todo} be a list containing only the all ones vector of length $|V|$. The vectors $\mathbf{a}$ in \texttt{todo} have columns indexed by rows of $V$ and represent partial partitions of the $n$ qubits. The value of $\mathbf{a}[\mathbf{v}]$ encodes the following information:
    \begin{itemize}
        \item $0$: $\text{supp}(\mathbf{v})$ is \text{not} a partition;
        \item $1$: whether $\text{supp}(\mathbf{v})$ is a partition has not yet been determined;
        \item $2$: $\text{supp}(\mathbf{v})$ is a partition. For depth-one operators, any $\mathbf{u}$ with overlapping support (i.e.\ $\mathbf{u}\cdot \mathbf{v} \ge 0$ is \text{not} a partition).
    \end{itemize}
    \item While \texttt{todo} is not empty:
    \begin{enumerate}
    \item Pop the vector $\mathbf{a}$ from the end of \texttt{todo};
        \item Reorder the columns of $\mathbf{z}$ and $K_L$ by moving the columns with $\mathbf{a}[\mathbf{v}] = 0$ to the far left, the columns with $\mathbf{a}[\mathbf{v}] = 1$ to the middle and the columns with $\mathbf{a}[\mathbf{v}] = 2$ to the far right. 
        \item Calculate $\mathbf{z}' := \text{Res}_{\mathbb{Z}_N}(K_L,\mathbf{z})$. If $\mathbf{z}'[\mathbf{v}] > 0$ for any $\mathbf{v}$ where $\mathbf{a}[\mathbf{v}] = 0$ then the partition is not valid. This is because taking the residue will eliminate the leftmost entries of $\mathbf{z}$ if possible by adding rows of $K_L$;
        \item If the partition is valid, find the the first $\mathbf{v}$ such that $\mathbf{z}'[\mathbf{v}] > 0$ and $\mathbf{a}[\mathbf{v}] = 1$;
        \item If there is no such $\mathbf{v}$, we have a depth-one implementation. Return the qubits to their original order and return $\mathbf{z}'$;
        \item Otherwise, let $\mathbf{a}_1$  be the same as $\mathbf{a}$ but where $\mathbf{a}[\mathbf{v}] = 2$ and $\mathbf{a}[\mathbf{u}] = 0$ for all $\mathbf{u}$ such that $\mathbf{u}\cdot \mathbf{v} > 0$. Let $\mathbf{a}_2$ be the same as $\mathbf{a}$ but with $\mathbf{a}[\mathbf{v}] = 2$. The vectors $\mathbf{a}_1$ and $\mathbf{a}_2$ represent the two possible scenarios where either $\mathbf{v}$ is or is not a partition - append them to \texttt{todo};
        \item Return to Step 3.
    \end{enumerate}  
    \item Return \texttt{FALSE} as all possible configurations have been explored.
\end{enumerate}
The above algorithm yields depth-one operators composed of physical phase-rotation gates. If implementations using physical controlled-phase gates are required, convert $\mathbf{z}$ and $K_L$ to controlled-phase representations using the method in \Cref{sec_duality}. If we require a logical operator with exactly the same action as the original operator with $Z$-component $\mathbf{z}$, substitute the $Z$-components of the diagonal logical identity generators $K_M$ of \Cref{sec_LI} for $K_L$. 

\subsection{Representation of  Controlled-Phase Operators as XP Operators via Embedding Operator}
In this Section we prove that the phase-rotation operators of \Cref{sec_phase_rotation} acting on a codespace correspond to diagonal XP operators in the embedded codespace defined in \Cref{sec_embedding_operators}. We do this by demonstrating that the mapping of phase-rotation operators acting on the codespace to XP operators in the embedded codespace of \Cref{sec_embedding_operators} is a group homomorphism. 

\begin{proposition}[Embedding operator induces a group homomorphism]\label{prop_emb_hom}
The embedding operator $\mathcal{E}_V$ defined as follows is a group homomorphism between $\mathcal{XRP}^V_N$ and $\mathcal{XP}_N^{|V|}$:
\begin{align}
    \mathcal{E}_V(\text{XRP}^V_N(p|\mathbf{x}|\mathbf{q})) := XP_N(p|\mathbf{x}V^T|\mathbf{q})
\end{align}
\end{proposition}
\begin{proof}
We prove this by considering generators of the group $\mathcal{XRP}^V_N$. Let $\mathbf{b}^n_i$ be the length $n$ binary vector which is zero apart from component $i$ which is one and consider $X_i, X_j$ for $0 \le i, j < n$:
\begin{align}
    \mathcal{E}_V(X_i X_j) &= \mathcal{E}_V\big(\text{XRP}_V(0|\mathbf{b}^n_i + \mathbf{b}^n_j|\mathbf{0})\big)\\
    &= XP_N(0|(\mathbf{b}^n_i + \mathbf{b}^n_j)V^T|\mathbf{0})\\
    &= XP_N(0|(\mathbf{b}^n_iV^T|\mathbf{0})XP_N(0|(\mathbf{b}^n_jV^T|\mathbf{0})\\
    &= \mathcal{E}_V(X_i)\mathcal{E}_V(X_j).
\end{align}
By a similar argument, $\mathcal{E}_V\big(\text{RP}_N(2,\mathbf{u})  \text{RP}_N(2,\mathbf{v})\big) = \mathcal{E}_V\big(\text{RP}_N(2,\mathbf{u})\big)\mathcal{E}_V\big(\text{RP}_N(2,\mathbf{v})\big)$ for $\mathbf{u}, \mathbf{v} \in V$.  Where $X$ operators precede diagonal operators we have:
\begin{align}
\mathcal{E}_V(X_i\text{RP}_N(2,\mathbf{v})) &= \mathcal{E}_V(\text{XRP}_V(0|\mathbf{b}^n_i|\mathbf{b}^{|V|}_\mathbf{v}))\\
&= XP_N(0|\mathbf{b}^n_iV^T|\mathbf{b}^{|V|}_\mathbf{v})\\
&= \mathcal{E}_V(X_i)\mathcal{E}_V\big(\text{RP}_N(2,\mathbf{v})\big)
\end{align}
Where diagonal operators precede X operators, we first consider the case where $\mathbf{v}[i]=0$. In this case, the operators commute so we can swap the order of operators so that the X operators precede the diagonal operator. Now consider the case $\mathbf{v}[i]=1$ where the operators do not commute:
\begin{align}
    \mathcal{E}_V\big(\text{RP}_N(2,\mathbf{v})X_i\big) &= \mathcal{E}_V\big(\omega^2 X_i\text{RP}_N(-2,\mathbf{v})\big)\\
    &= XP_N(2|\mathbf{b}^n_iV^T|-\mathbf{b}^{|V|}_\mathbf{v})
\end{align}
Due to the commutation relation of \Cref{eq_XP_relation} and because $(\mathbf{b}^n_iV^T) \mathbf{b}^{|V|}_\mathbf{v} = \mathbf{b}^{|V|}_\mathbf{v}$ when $
\mathbf{v}[i] = 1$:
\begin{align}
    \mathcal{E}_V\big(\text{RP}_N(2,\mathbf{v})\big)\mathcal{E}_V(X_i) &= XP_N(0|\mathbf{0}|\mathbf{b}^{|V|}_\mathbf{v})XP_N(0|\mathbf{b}^n_iV^T|\mathbf{0})\\
    &= XP_N(2|\mathbf{b}^n_iV^T|-\mathbf{b}^{|V|}_\mathbf{v})\\
    &= \mathcal{E}_V(\text{RP}_N(2,\mathbf{v})X_i)
\end{align}
Because group operations are preserved for generators of the group, the embedding is a group homomorphism.
\end{proof}
\subsection{Canonical Form of Logical Phase-Rotation Operators}
In this Section, we show that the canonical form of logical phase-rotation operators discussed in  \Cref{sec_canonical_logical_operators} acts as a logical operator as claimed.
\begin{proposition}[Logical Phase Rotation Operator]\label{prop_logical_RP}
Let $L_Z$ be a binary matrix representing the $Z$-components of logical Z operators such that $L_Z^TL_X = I_k$ and let $\mathbf{w}$ be a binary vector of length $k$.

The operator $\text{RP}_N(2,\mathbf{w}L_Z)$ acts as a logical $\overline{\text{RP}_N(2,\mathbf{w})}$ operator.
\end{proposition}
\begin{proof}
    This can be seen by considering the action of the operator on the computational basis element $|\mathbf{e}_\mathbf{uv}\rangle$ where $\mathbf{e}_\mathbf{uv} := \mathbf{u}S_X + \mathbf{v}L_X$. From the argument in \Cref{prop_logical_P},  $\mathbf{e}_\mathbf{uv} \cdot \mathbf{z}_j \mod 2 = 1 \iff \mathbf{v}[j] = 1$. Hence:
\begin{align}
    \mathbf{e}_\mathbf{uv} \cdot (\bigoplus_{j \preccurlyeq \mathbf{w}}\mathbf{z}_j) \mod 2 = 1 &\iff \bigoplus_{j \preccurlyeq \mathbf{w}} \mathbf{v}[j] = 1 \\&\iff \mathbf{v}\cdot \mathbf{w} \mod 2 = 1 
\end{align}
Hence, the phases applied by the operators are the same and the result follows.
\end{proof}

\section{Application of Methods to Non-CSS Stabiliser Codes}\label{app_non_css}
In this work, we have focused on identifying diagonal logical operators for CSS codes in the form defined in \Cref{sec_css_codes}. In this Section, we show how to find diagonal logical operators for arbitrary non-CSS stabiliser codes. We will prove the following main proposition:
\begin{proposition}[Mapping Non-CSS Stabiliser Codes to CSS Codes]\label{prop_non_css}
    Let $\mathbf{C}$ be the codespace of a Pauli stabiliser code on $n$ qubits. There exists a CSS code on $n$ qubits with codespace $\mathbf{C}'$ such that $\mathbf{C} = DQ \mathbf{C}'$ where $Q:= XP_2(0|\mathbf{q}|\mathbf{0})$, $\mathbf{q}$ is a length $n$ binary vector and $D$ is a diagonal level 2 Clifford operator. Furthermore, a diagonal operator  $\overline{B}$ is a logical $B$ operator of $\mathbf{C}'$ if and only if $Q\overline{B}Q^{-1}$ is a logical $B$ operator of $\mathbf{C}$.
\end{proposition}
The CSS code $\mathbf{C}'$ in \Cref{prop_non_css} may have different error correction properties to $\mathbf{C}$ (i.e. weight of stabiliser generators and logical operators), but allows us to determine the diagonal logical operator structure of $\mathbf{C}$.
In this Section, we first introduce some background material on non-CSS stabiliser codes. CSS codes of the form of \Cref{sec_css_codes} have diagonal stabiliser generators with zero phase components and non-diagonal stabiliser generators with zero phase and $Z$-components. This is not the case for arbitrary stabiliser codes, and we show how to eliminate these components in two steps to yield the operators $Q$ and $D$ in the above proposition.  We illustrate \Cref{prop_non_css} by applying it to the perfect 5-qubit code of Ref.~\cite{5_qubit_code}.

\subsection{Background on Non-CSS Codes}
Arbitrary Pauli stabiliser codes have stabiliser generators from the Pauli group $\langle iI, X, Z\rangle^{\otimes n} = \mathcal{XP}_2^n$. 
A method of determining a canonical set of independent stabiliser generators, logical $X$ and logical $Z$ operators is given on  page 477 of Ref.~\cite{nielsen_chuang_2010}. 
Let $\mathbf{S}_X$ and $\mathbf{S}_Z$ be the canonical stabiliser generators and let $\mathbf{L}_X$ be the  canonical logical $X$ operators.
Elements of $\mathbf{S}_Z$ may have signs of $\pm 1$ and elements of $\mathbf{S}_X$ may have non-trivial phase and $Z$-components.
For \Cref{prop_non_css}, we require that $\mathbf{C}' := (DQ)^{-1}\mathbf{C}$ is stabilised by diagonal generators with trivial phase components and non-diagonal  generators with trivial phase and $Z$-components.

We now set out a canonical form for the codewords of the stabiliser code $\mathbf{C}$. 
Let $r$ be the number of operators in $\mathbf{S}_X$ and $k$ the number of operators in $\mathbf{L}_X$ and let $\mathbf{v} \in \mathbb{Z}_2^k$.  
Let $\mathbf{q}$ be a binary vector of length $n$ such that $B\ket{\mathbf{q}} = \ket{\mathbf{q}}$ for all $B \in \mathbf{S}_Z$.
Define $\mathbf{L}_X^\mathbf{v}:= \prod_{i \preccurlyeq \mathbf{v}}\mathbf{L}_X[i]$ where $\mathbf{L}_X[i]$ is the $i$th operator in $\mathbf{L}_X$. 
Due to the arguments in Sections 4.2 and 6.2 of \cite{XP}, the following codewords span the codespace $\mathbf{C}$ and define the encoding map $C$ of $\mathbf{C}$ (\Cref{sec_LO_of_CSS_codes}):
\begin{align}
    \ket{\mathbf{v}}_L = \sum_{\mathbf{u}\in \mathbb{Z}_2^r}\mathbf{S}_X^{\mathbf{u}}\mathbf{L}_X^{\mathbf{v}}\ket{\mathbf{q}}.\label{eq_non_css_codewords}
\end{align}

We now discuss how the codewords and logical operators of a stabiliser code $\mathbf{C}$ transform when the codespace is acted upon by a unitary operator $U$. The codewords of the transformed code $\mathbf{C}' := U\mathbf{C}$ are given by $U\ket{\mathbf{v}}_L$ so the encoding map of $\mathbf{C}'$ is given by $U\mathcal{C}$. The operator $A$ is a logical identity of $\mathbf{C}$ if and only if $UAU^{-1}$ is a logical identity of $\mathbf{C}'$. This is because:
\begin{align}(UAU^{-1})U\ket{\mathbf{v}}_L = UA\ket{\mathbf{v}}_L = U\ket{\mathbf{v}}_L.\end{align}
As the stabiliser generators $\mathbf{S}_X$ and $\mathbf{S}_Z$ are elements of the logical identity group, they also update via conjugation.
The operator $\overline{B}$ is a logical $B$ operator on $\mathbf{C}$ if and only if $U\overline{B}U^{-1}$ is a logical $B$ operator on $\mathbf{C}'$ because for all logical identities $A$ of $\mathbf{C}$ the requirements of \Cref{sec_LO_test} and \Cref{sec_LO_of_CSS_codes} are met as follows:
\begin{align}
[[ U\overline{B}U^{-1}, U A U^{-1} ]] &= U[[\overline{B},A]]U^{-1}; \text{ and}\\
(U\overline{B}U^{-1})U\mathcal{C} &= U \overline{B} \mathcal{C} = (U\mathcal{C})B.
\end{align}

\subsection{Eliminating Phase Components from Diagonal Stabiliser Generators}
We now show how to find the vector $\mathbf{q}$ in \Cref{eq_non_css_codewords} which allow us to eliminate signs from the diagonal stabiliser generators of the non-CSS code $\mathbf{C}$ via conjugation by the operator $Q:=XP_2(0|\mathbf{q}|\mathbf{0})$. 

The canonical diagonal stabiliser generators $\mathbf{S}_Z$ are of form $XP_2(2p_i|\mathbf{0}|\mathbf{z}_i)$ where $p_i \in \mathbb{Z}_2$ and $\mathbf{z}_i\in\mathbb{Z}_2^n$. Let $E_s$ be the binary matrix with rows of form $(p_i|\mathbf{z}_i)$ and let $K_s := \ker_{\mathbb{Z}_2}(E_s)$. If $p_i = 1$ for any $i$, the top row of $K_s$ is of form $(1|\mathbf{q})$ and satisfies $p_i + \mathbf{q}\cdot \mathbf{z}_i = 0 \mod 2$ for all $i$. 
Now let $Q:= XP_2(0|\mathbf{q}|\mathbf{0})$ then we also have $Q XP_2(2p_i|\mathbf{0}|\mathbf{z}_i) Q = XP_2(2p_i + 2\mathbf{q}\cdot\mathbf{z}_i|\mathbf{0}|\mathbf{z}_i) = XP_2(0|\mathbf{0}|\mathbf{z}_i)$. Hence conjugation by $Q$ eliminates the phase components of the diagonal stabiliser generators as required. 
As $Q$ is non-diagonal, the diagonal logical operators and identities may update on conjugation by $Q$.

\subsection{Eliminating Phase and $Z$-Components from Non-Diagonal Stabiliser Generators}
We now show how to find a diagonal level 2 Clifford operator $D$ from \Cref{prop_non_css} which allows us to eliminate the phase and $Z$-components of the non-diagonal stabilisers $\mathbf{S}_X$. 
Let $\ket{S}$ be the state stabilised by the set of $n$ independent operators $\mathbf{S}_X, \mathbf{S}_Z$ and $ \mathbf{L}_X$. 
We can write $\ket{S}$ as follows:
\begin{align}
    \ket{S} = \sum_{\mathbf{u}\in \mathbb{Z}_2^r,\mathbf{v}\in \mathbb{Z}_2^k}\mathbf{S}_X^{\mathbf{u}}\mathbf{L}_X^{\mathbf{v}}\ket{\mathbf{q}} = \sum_{\mathbf{v}} \ket{\mathbf{v}}_L.
\end{align}       
Let $S_X$ and $L_X$ be the binary matrices formed from the X-components of $\mathbf{S}_X$ and $\mathbf{L}_X$ respectively. Using the terminology of Proposition 5.1 of Ref.~\cite{XP}, $\ket{S}$ is an XP state of precision $N=2$ and so is a weighted hypergraph state of form:
\begin{align}
    \ket{S} = D\sum_{\mathbf{u,v}}\ket{\mathbf{u}S_X + \mathbf{v}L_X + \mathbf{q}}= DQ\sum_{\mathbf{u,v}}\ket{\mathbf{u}S_X + \mathbf{v}L_X}.
\end{align}
The operator $D$ is a product of diagonal level 2 Clifford operators and can calculated via the method in Algorithm 5.3.1 of Ref.~\cite{XP}. 
Now let $\mathbf{C}'$ be the CSS code specified by the X-checks $S_X$ and X-logicals $L_X$. Due to \Cref{sec_css_codes}, codewords of $\mathbf{C}'$ are of form $\ket{\mathbf{v}}'_L:= \sum_{\mathbf{u}}\ket{\mathbf{u}S_X + \mathbf{v}L_X}$ and so the codewords of $\mathbf{C}$ can be written:
\begin{align}
    \ket{\mathbf{v}}_L = DQ\sum_{\mathbf{u}}\ket{\mathbf{u}S_X + \mathbf{v}L_X} =  DQ\ket{\mathbf{v}}'_L.
\end{align}
Hence, $\mathbf{C} = DQ\mathbf{C}'$ as required. Transforming a CSS code $\mathbf{C}'$ by the diagonal operator $D$ has no effect on the diagonal stabiliser generators, logical identities or logical operators. However, it can increase the weight of non-diagonal stabiliser generators and logical X operators, and so increase the code distance.
\begin{example}[Perfect 5-Qubit Code]
Let $\mathbf{C}$ be the perfect 5-qubit code of \cite{5_qubit_code} with stabiliser generators and logical $X$ operator as follows:
\begin{align}
    \mathbf{S} &:= \begin{pmatrix}
        XZZXI\\
        IXZZX\\
        XIXZZ\\
        ZXIXZ
    \end{pmatrix};&
    \overline{X} &:= ZIIZX
    \end{align}.
Let $\mathbf{C}'$ be the CSS code with $X$-checks and $X$-logicals:
\begin{align}
    S_X &:= \begin{pmatrix}
        10010\\
        01001\\
        10100\\
        01010
    \end{pmatrix};& 
    L_X &:= \begin{pmatrix}
        00001
    \end{pmatrix}.
\end{align}
We find that $D = CZ_{01}CZ_{12}CZ_{23}CZ_{01}CZ_{34}CZ_{40}$ satisfies $\mathbf{C} = D\mathbf{C}'$ using the conjugation rule $CZ_{01} X_0 CZ_{01} = X_0 Z_1$. Whilst $\mathbf{C}$ has distance 3, $\mathbf{C}'$ has distance 1. In the \href{https://github.com/m-webster/CSSLO/blob/main/notebooks/10.3_non-CSS.ipynb}{linked Jupyter notebook}, users can use the above method to find $D,Q$ and $\mathbf{C}'$  for various non-CSS stabiliser codes from \href{http://www.codetables.de}{www.codetables.de}.
\end{example}

\medskip

\bibliographystyle{unsrt}

\end{document}